\pdfoutput=1
\documentclass[paper=letter, fontsize=10pt,leqno]{scrartcl}
\usepackage[body={5.7in,7.5in}]{geometry}
\usepackage[T1]{fontenc}

\usepackage[english]{babel}
\usepackage{csquotes}

\usepackage[protrusion=true,expansion=true]{microtype}				
\usepackage{amsmath}										
\usepackage{amsthm}
\usepackage{latexsym}
\usepackage{hyperref}

\usepackage[pdftex]{graphicx}
\usepackage{url}
\usepackage{amssymb}
\usepackage{bbm}
\usepackage{MnSymbol}

\usepackage[pdftex]{color}
\usepackage{eucal}
\usepackage{amscd}
 \usepackage{verbatim}
 \usepackage{bigints}
 \usepackage[format=plain]{caption}

\usepackage{epstopdf}
\usepackage{wrapfig}

\usepackage{sectsty}					  
\allsectionsfont{\centering \normalfont\scshape}	  

\usepackage{mathtools}

\usepackage{fancyhdr}
\newtheorem{theorem}{Theorem}

\newtheorem{proposition}[theorem]{Proposition}

\theoremstyle{definition}
\newtheorem{definition}[theorem]{Definition}
\newtheorem{assumption}[theorem]{Assumption}
\newtheorem*{remark*}{Remark}

\title{
		\vspace{-1in} 	
		\normalfont \normalsize \textsc{} \\ [0pt]
		\huge Entropy Production in  Random Billiards
}

\author{\normalfont \large 
  T. Chumley\footnote{Department of Mathematics and Statistics, Mount Holyoke College, 50 College St, South Hadley, MA 01075}, 
\ R. Feres\footnote{Department of Mathematics and Statistics, Washington University, Campus Box 1146, St. Louis, MO 63130}
}

\begin{document}


\maketitle

\begin{abstract}
\begin{center}
Abstract 
\end{center}

\medskip

{\small We consider a class of random mechanical systems called {\em random billiards} to study the problem of quantifying the irreversibility of nonequilibrium macroscopic systems.  In a random billiard model, a point particle evolves by free motion through the interior of a spatial domain, and reflects according to a reflection operator, specified in the model by a Markov transition kernel, upon collision with the boundary of the domain.  We derive a formula for entropy production rate that applies to a general class of random billiard systems.  This formula establishes a relation between the purely mathematical concept of entropy production rate and textbook thermodynamic entropy, recovering in particular Clausius' formulation of the second law of thermodynamics.  We also study an explicit class of examples whose reflection operator, referred to as the Maxwell-Smoluchowski thermostat, models systems with boundary thermostats kept at possibly different temperatures.  We prove that, under certain mild regularity conditions, the class of models are uniformly ergodic Markov chains and derive formulas for the stationary distribution and entropy production rate in terms of geometric and thermodynamic parameters.}

\end{abstract}

\section{Introduction}

\emph{Model and results.} The overarching focus of this paper is to study the dynamics and thermodynamic properties of a class of random mechanical systems, referred to as \emph{random billiards}, that serve as a concrete model for the rigorous, analytic study of nonequilibrium phenomena.  Random billiards can be seen as a Markov chain model variation on mathematical billiards with a random reflection in place of the classical law of specular reflection.  The random reflections are specified by a choice of Markov transition operator, which we call the {\em random reflection operator}, defined on the space of post-collision velocities.

The main results of the paper are summarized as follows, with more detailed and explicit statements to come in the next section.  We use the notion entropy production rate, defined in terms of the Kullback-Leibler divergence of a general stochastic process, to give a characterization of the irreversibility of random billiard systems in terms of the thermodynamic and geometric parameters that characterize the Markov transition operator.  In particular, we give an explicit formula for the entropy production rate in terms of the stationary distribution of our Markov chain model and show that the entropy production rate is always a non-negative quantity.  The result, which can be seen as the second law of thermodynamics for random billiard models, lays the groundwork for a more detailed study of how certain aspects of the models, such as the geometry of the billiard table or the choice of random reflection operator, affect the irreversibility of a system.  In the second half of the paper, we study random billiard systems with a random reflection operator called the {\em Maxwell-Smoluchowski reflection model} which models a boundary thermostat, or external heat bath, that keeps the system out of equilibrium.  We show that random billiard systems with the Maxwell-Smoluchowski reflection model are uniformly ergodic Markov chains, give an explicit description of the stationary distribution, and compute the entropy production rate analytically and numerically for various examples.

\emph{Main technical issues.} We aim to formalize problems in nonequilibrium statistical physics which are related to the second law of thermodynamics.  In particular, a primary problem is to quantify the notion of irreversibility of macroscopic systems, which are defined by reversible microscopic behavior, using entropy production rate. This has been done for a large class of mesoscopic models, namely stochastic differential equations (SDE), and a large class of microscopic models, namely countable state Markov chains; see the book \cite{qian} for a comprehensive overview. However, as far as we are aware, a rigorous study of entropy production rate for continuous state Markov chains, particularly those with non-compact state space, has not been done. The present work addresses this gap in the literature in the specific context of random billiard Markov chains, but the techniques and statements of results should hold for more general non-compact state space Markov chain models. The main work in establishing the entropy production rate formula and subsequent second law of thermodynamics is in showing that the probability measure for the Markov chain in forward time is mutually absolutely continuous with respect to the time reversed measure, from which it follows through careful calculations that the entropy production rate is well defined and satisfies a formula that makes explicit its dependence on the steady state stationary measure of the system. The work here requires a bit more detailed analysis than the SDE case because we cannot use the Cameron-Martin-Girsanov formula that is available for diffusion processes.

In the second part of the paper, we use coupling in order to prove the uniform ergodicity of random billiard models with the Maxwell-Smoluchowski reflection model. Such techniques are now standard, but there are always context specific details that need to be addressed, particularly in the case of non-compact state spaces. It should also be emphasized here that our model is particularly amenable to exact analytic study. Beyond proving uniform ergodicity, we are able give an explicit formula for the stationary measure, even in the case when the boundary thermostats are kept at different temperatures. The formula is given both in terms of temperatures and the geometry of the billiard table of the system. We believe that the techniques here should also serve as a basis for the further study of random billiards with boundary thermostats modeled by other reflection laws.

\emph{Related work.} The work in this paper fits into several established, overlapping lines of research\----random billiard models, models of heat conduction, and the mathematical analysis of nonequilibrium phenomena through entropy production rate. These lines of research  are united by the theme of studying models that are simple enough to be amenable to rigorous analysis but rich enough to demonstrate transport phenomena. While this philosophy is taken by much of the related work to be discussed below, the present work is distinguished in the way that explicit formulas linking entropy production rate, temperature, and microscopic parameters defining the system are attained. 

In random billiard models, much work has been done in studying the ergodicity of models with independent, identically distributed (i.i.d.) reflections. In \cite{E2001}, uniform ergodicity is proved for random billiards with the Knudsen reflection model that fixes the speed of the particle on tables with $C^1$ boundary. This work was extended in \cite{CPSV2009} to more general, but fixed speed reflection models, along with more relaxed restrictions on the boundary of the billiard domain. The issue of ergodicity for non-i.i.d. models, derived from Markov operators that model more complex boundary microstructures is arguably more subtle. Along these lines, ergodicity and rate of approach to equilibrium have been studied for certain classes of boundary microstructure in \cite{Feres2007, FZ2010, FZ2012}. For reflection models that model boundary thermostats less is known. In \cite{cook}, it is shown that the Maxwell-Boltzmann distribution is the equilibrium distribution for a large class of examples with a single heat bath. Ergodicity was proved for reflection models with multiple temperatures in \cite{khanin-yarmola}, but this is for a special class of billiard tables, using the geometry of dispersing Sinai billiards.

While there has been a resurgence of interest in the recent past for mechanical models of heat conduction and the study of nonequilibrium behavior in systems with thermostatted boundaries, there are few results that give explicit descriptions of the steady state distribution of the system and few results that give analytic, functional relationships between thermodynamic quantities like the entropy production rate to parameters of the system beyond thermostat temperatures. We should note that in the case of anharmonic oscillators driven by stochastic heat baths, there is in fact a fairly complete set of results. Under certain general assumptions, they can be modeled by stochastic differential equations. There are results on the existence and uniqueness of stationary distributions \cite{MR1685893} and rate of convergence, as well as other statistical properties, are well understood \cite{MR1889227, MR2227858}. Moreover, techniques as in \cite{qian} can be used to characterize and study entropy production rate \cite{MR1915300}. On the other hand, these are results at the mesoscopic scale, but things are much less complete for models defined at the microscopic scale, which include random billiards, from which mesoscopic models should arise through universal limiting laws. In addition to random billiard models, other microscopic models include those based on purely Hamiltonian mechanical systems coupled to stochastic heat baths \cite{lin, collet-eckmann, larralde, lin, MR2200889, MR3217533} and hybrid models consisting of mechanical systems with a stochastic component meant to make the systems more tractable \cite{MR3101481}. While much of the work on these models has been in studying the stationary states and rate of convergence to stationarity, the study of entropy production is a bit more limited. In \cite{bellet-demers-zhang}, fluctuations of entropy production rate are studied for perturbed Lorentz gas. Surveys such as \cite{MR3517572, MR1705589, jpr, hongqian, seifert} study entropy production for general classes of dynamical systems and Markov processes, but this approach is a bit different than the one we take, where we are concerned with how  explicit, microscopic details of the model affect entropy production rate.

\emph{Organization of the paper.} Section \ref{sec:main} introduces definitions and notation, and concludes with more precise statements of main results.  Section \ref{sec:Basic facts} establishes basic facts about random billiard Markov chains and random reflection operators.  Section \ref{sec:Second Law of Thermodynamics} establishes the preliminary details for the Second Law of Thermodynamics for random billiards and gives a proof of this result.  Section \ref{sec:Multi-temperature Maxwell-Smoluchowski systems} introduces random billiards with the Maxwell-Smoluchowski thermostat model.  There, we prove the uniform ergodicity of such Markov chains, giving an explicit formula for the stationary distribution, present an explicit formula for entropy production rate, and present some analytical and numerical results for some representative examples.  We conclude with Section \ref{sec:Efficiency and work production} where we give a final example which shows that random billiard systems can be used to produce work against an external force, creating what we call random billiard heat engines.  We present a short numerical study of efficiency and work production for random billiard heat engines.

\section{Definitions and Main Results}\label{sec:main}

Let $M$ be a smooth manifold of dimension $n$ with boundary $\partial M$. 
The boundary may contain points of non-differentiability, where the tangent space to $\partial M$ is not defined. 
Points where the tangent space is defined will be called {\em regular}.
The notion of a {\em manifold with corners} as defined in \cite{lee} is general enough to include all the examples of interest to us.
As the issue of regularity is not central  to the  results of this paper, we  simply assume throughout  that all points under consideration are  
regular.

Let $\pi:TM\rightarrow M$ indicate the tangent bundle to $M$.
The notations $x=(q,v)\in TM$ and $v\in T_qM$ will be used.
We are mostly concerned with tangent vectors at boundary points.
Thus it makes sense to introduce the bundle
$$\pi:N=\{(q,v)\in TM:q\in \partial M\}\rightarrow M $$ 
corresponding to the restriction of $\pi$ to the boundary.
Let  $N_q:=N\cap T_qM$.
We equip $M$ with a Riemannian metric $\langle\cdot,\cdot\rangle$.  
The inward pointing unit normal vector to the boundary at a regular point $q$ will be written $\mathbbm{n}_q$. 
Vectors $v\in N_q$ pointing to the interior of $M$ constitute the set $N_q^+$ of {\em post-collision} velocities; the negative of vectors in $N_q^+$ constitute the set $N_q^-$ of {\em pre-collision} velocities.  
Thus $$N_q^\pm := \{v\in N_q: \pm \langle v,\mathbbm{n}_q\rangle_q\geq 0\}. $$
The disjoint union of the $N_q^\pm$ over the regular boundary points defines $N^\pm$.

We suppose that $\partial M$ has finite $(n-1)$-dimensional volume and denote by $dA(q)$ the Riemannian volume element  of $\partial M$ at $q$. 
The notation $\bar{A}=A/A(\partial M)$ will indicate the normalized volume  when $\partial M$ has finite volume. 
Let $t\mapsto \varphi_t(x)$ denote the geodesic flow in $TM$, which is only defined for values of $t$ up to the moment when geodesics reach the boundary. 
Recall that the geodesic flow is the Hamiltonian flow (on the tangent bundle) for the free motion of a particle of mass $m$ with kinetic energy $E(q,v)=\frac12 m |v|_q^2$, 
where $|v|_q^2=\langle v,v\rangle_q$. 
For $x=(q,v)\in N^+$, let $t(x):= \inf \{t > 0 : \pi(\varphi_t(x))\in \partial M\}$.
The {\em return map} (to the boundary) $\mathcal{T}:N^+\rightarrow N^-$  is defined as 
$$ \mathcal{T}(x):=\varphi_{t(x)}(x).$$
Upon reaching the boundary, the billiard trajectory (i.e., an orbit of the geodesic flow) is reflected back into the manifold by a choice of map from $N^-$ to $N^+$; 
in deterministic billiards, the standard choice is the specular reflection $v\mapsto v-2\langle v,\mathbbm{n}_q\rangle_q \mathbbm{n}_q$. (The  theory of standard billiard dynamical systems largely deals with planar systems, as in  \cite{chernov}, but see also \cite{gaspard}.) 
For random billiards, this is done by means of a choice of {\em reflection operator} $P_q$ at each boundary point $q$, to be defined shortly.

It is possible to make this set-up more general by adding conservative forces, in which case the geodesic flow should be replaced with the Hamiltonian flow corresponding to a choice of potential function $\Phi:M\rightarrow \mathbb{R}$. Most of the paper will be concerned with free motion ($\Phi=0$) between collisions, except in Theorem \ref{heat entropy} (the Second Law of Thermodynamics), where we assume the more general setting. We then indicate by $E_0$ the above kinetic energy function and use $E$ to denote the total energy $E= E_0+\Phi\circ\pi$. The introduction of potential forces does not affect the form of the reflection law, which is intended to model impulsive forces; that is, very strong forces acting on very short time intervals. 

The definition of reflection operators given next is motivated by natural boundary  conditions for the Boltzmann equation  involving gas-surface interaction. See for example Chapter 1 in \cite{cercignani}. 

The space of Borel probability measures on a topological space $X$ will be indicated by $\mathcal{P}(X)$. Two probability measures will be called equivalent if they are mutually absolutely continuous. 
We often denote by $\mu(f):=\int_X f(x)\, d\mu(x)$ the integral of a function $f$ on $X$ with respect to a measure $\mu$.
A measure $\mu_\lambda$ will be said to depend measurably on elements   $\lambda$ of a measurable space if   for any bounded continuous function
$f$
 on $X$ the map
$\lambda\mapsto \mu_\lambda(f)$ is measurable.

\begin{definition}[General reflection operator]
A general  {\em reflection operator} at a regular boundary point $q$ is
a map $v\in N_q^- \mapsto P_{(q,v)}\in  \mathcal{P}(N_q^+)$.  A general reflection operator $P$ on $M$ is the assignment of such an operator for each regular boundary point $q$. 
 We suppose that $P$ depends measurably on $x=(q,v)$ in the sense that for any given bounded continuous function $f$ on $N^+$, the map $x\mapsto P(f)(x):= P_x(f)$ is measurable.
\end{definition}

The  operator $P$ on $M$ gives rise at each $q$ to a map $P_q$ from $\mathcal{P}(N_q^-)$ to $\mathcal{P}(N_q^+)$ defined by $\nu\mapsto \nu P_q$, where the integral of a test function $f$ on $N_q^+$ with respect to the latter is
$$(\nu P_q)(f) = \int_{N_q^-} P(f)(x)\, d\nu(x).  $$

A reflection operator will be defined as a general reflection operator satisfying the condition of {\em reciprocity}, 
whose definition depends on   the notion of a {\em Maxwellian} probability measure defined below.  
It is through the notion of reciprocity that  the property of $\partial M$ having a  given, fixed, temperature at a boundary point will be defined.
Let $dV_q(v)$ denote the Riemannian volume measure on $T_qM$.  
A measure $\mu$ on $N_q$ is said to have {\em density} $\rho(v)$ if $d\mu(v) = \rho(v)\, dV_q(v)$.

\begin{definition}[Maxwellian at temperature $T$]
The {\em Maxwellian} (or Maxwell-Boltzmann probability distribution) at  boundary point $q\in M$ and temperature $T(q)$ is the probability measure $\mu^\pm_q\in \mathcal{P}(N_q^\pm)$ having density
\begin{equation}\label{maxwellian} \rho_q(v) = 2\pi\left(\frac{\beta(q)m}{2\pi}\right)^{\frac{n+1}{2}}|\langle v, \mathbbm{n}_q\rangle| \exp\left\{-\beta(q)\frac{m|v|_q^2}{2}\right\}\end{equation}
where $m$ is the mass of the billiard particle, $\beta(q)=1/\kappa T(q)$, and  $\kappa$ is known as the Boltzmann constant.  
\end{definition}

At each regular boundary point $q$ consider the following  linear involutions: 
the {\em flip map}  $$J:N\rightarrow N, \ \ \ J(q,v)=(q, J_qv)=(q,-v)$$ and  the {\em time reversal} map
 $$\mathcal{R}_q(u,v)=(J_qv, J_qu)$$ on $N_q^-\times N_q^+$. Given a general reflection operator $P$, define $\zeta_q\in \mathcal{P}(N^-_q\times N^+_q)$ as
$$d\zeta_q(u,v)=d\mu^-_q(u)\, dP_{(q,u)}(v),$$
where $\mu_q^-$ is a Maxwellian at $q$.  If $F:X\rightarrow Y$ is a measurable map between measure spaces and $\mu$ is a probability measure on $X$, the {\em push-forward} of $\mu$ under $F$ is the probability measure $F_*\mu$ on $Y$ defined by 
$(F_*\mu)(E):=\mu(F^{-1}(E))$, where $E$ is a measurable subset of $Y$. If $F$ is a self-map of $X$, then $\mu$ is said to be {\em invariant} under $F$ if  $F_*\mu=\mu$. 

\begin{definition}[Reciprocity]\label{reciprocity} The general reflection operator $P$ has the property of {\em reciprocity} if at each regular boundary point $q$ 
the probability measure $\zeta_q$ (just defined above) is invariant under the time-reversal map $\mathcal{R}_q$. A general reflection operator satisfying reciprocity will be called simply a {\em reflection operator}.
\end{definition}

The notion of reciprocity may be  interpreted as a local detailed thermal equilibrium of the boundary at $q$.  It says that if a particle of mass $m$ hits the boundary at $q$ with a random pre-collision velocity distributed according to the Maxwellian at temperature $T(q)$, then it is reflected with the same distribution (at the same temperature), and this random scattering process is time-reversible in the stochastic sense.  
A more general, and somewhat more technical, definition of reciprocity than that of Definition \ref{reciprocity} will be formulated at the beginning of  Section \ref{sec:Second Law of Thermodynamics}. Theorem \ref{heat entropy}, the second law of thermodynamics, will be
proved for random billiards satisfying this more general notion. 

The billiard map of a standard deterministic billiard system is the composition of 
the geodesic translation $\mathcal{T}: x\in N^+\mapsto \mathcal{T}(x)\in N^-$ defined above 
and specular reflection at $\pi(\mathcal{T}(x))$. 
In a random billiard system the second map is replaced with the random scattering post-collision velocity distributed according to $P_{\mathcal{T}(x)}$. 
See Figure \ref{random_billiard_map}. 
Define
$$ \mathcal{D}:=\{(x,y)\in N^+\times N^+: \pi(y)=\pi(\mathcal{T}(x))\}.$$
Then an orbit of the random billiard system is a sequence $\dots, x_{-1}, x_0, x_1, \dots $ 
for which every pair $(x_i,x_{i+1})$ belongs to  $\mathcal{D}$.  
If the random billiard map sends $x$ to $y$, 
the probability distribution of $y$ is given by $ \mathcal{B}_x:=P_{\mathcal{T}(x)}$.
In what follows, it will be convenient to  refer to   the map $x\mapsto \mathcal{B}_x$ itself     as the billiard map.

\begin{figure}
\begin{center}
\includegraphics[width=0.45\textwidth]{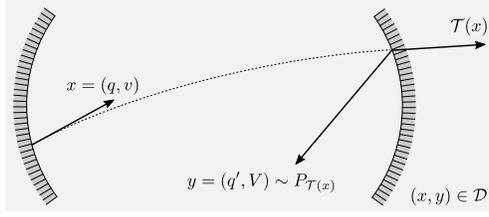}\ \ 
\caption{{\small The random billiard map is the composition of two maps: the geodesic translation $\mathcal{T}$ and scattering determined by  the reflection operator $P$. The distribution of the velocity $V$ after reflection is given by $\mathcal{B}_x=P_{\mathcal{T}(x)}$. }}
\label{random_billiard_map}
\end{center}
\end{figure}

\begin{definition}[Random billiard map]
The map $\mathcal{B}:N^+\rightarrow \mathcal{P}(N^+)$ defined by $\mathcal{B}_x=P_{\mathcal{T}(x)}$, where $P$ is a reflection operator, will be called the {\em random billiard map} associated to $P$. 
\end{definition}

We may use at different places the various notations
$$ (\mathcal{B}f)(x)=\mathcal{B}_x(f)=(\delta_x\mathcal{B})(f)=P_{\mathcal{T}(x)}(f)$$
where $\delta_x$ is the point mass supported at $x$ and $f$ is a test function on $N^+$.

Given a choice of initial probability distribution, we obtain from $\mathcal{B}$ a Markov chain $X_0, X_1, \dots$ on the state space $N^+$. 
Note that 
$$(\mathcal{B}f)(x)= \mathbb{E}\left[f(X_{i+1})|X_i=x\right] $$
where $\mathbb{E}$ indicates (conditional) expectation. 
We define the space of finite chain segments of length $n+1$ by
$$ \mathcal{D}_{[0,n]}=\{(x_0, \dots, x_n): (x_i,x_{i+1})\in \mathcal{D}, i=0, \dots, n-1\}.$$
Note that $\mathcal{D}=\mathcal{D}_{[0,1]}$.  
The notion of entropy production rate, to be considered shortly, involves the time reversal of the process. 
Clearly, simple reversal of the Markov chain, 
in which a chain segment $(x_0, \dots, x_n)$ is mapped to $(x_n,\cdots, x_0)$, 
cannot correspond to the physical idea of time reversal\----the direction of velocities must also be reversed. 
In order to define  {\em proper time reversal}  we introduce the map $\mathcal{J}=J\circ \mathcal{T}:N^+\rightarrow N^+$.

\begin{definition}[Proper time reversal map]
The {\em proper time reversal map}, or simply the {\em reversal map}, on chain segments is the map $ \mathcal{R}:\mathcal{D}_{[0,n]}\rightarrow \mathcal{D}_{[0,n]}$ defined by
$$\mathcal{R}: (x_0, \dots, x_n)\mapsto (\mathcal{J}x_n, \dots, \mathcal{J}x_0). $$
\end{definition}

A probability measure $\nu$ on $N^+$ is {\em stationary} for the  random billiard process if $\nu\mathcal{B}=\nu$.  Applied to a test function $f$ on $N^+$, this condition means that
$$ \int_{N^+}f(x)\, d\nu(x) = \int_{N^+} f(y)\,  dP_{\mathcal{T}(x)}(y)\, d\nu(x).$$
 It is useful to also define $\nu^-:=\mathcal{T}_*\nu\in \mathcal{P}(N^-)$. For emphasis we may on occasion write $\nu^+:=\nu$.
 The issue of existence, uniqueness, regularity, and ergodicity of stationary measures will be addressed in Section \ref{sec:Multi-temperature Maxwell-Smoluchowski systems} for a general class of examples.
 
 Let $X_0, X_1, \dots$ be a stationary Markov chain with state space $N^+$, transition operator $\mathcal{B}$, and stationary probability $\nu$.  For a test function $f$ on $N^+$,  $\mathbb{E}\left[f(X_i)\right]=\nu(f).$ Finite chain segments in $\mathcal{D}_{[0,n]}$ are distributed according to the probability measure $\mathbb{P}_{[0,n]}$ defined by
$$d\mathbb{P}_{[0,n]}(x_0, \dots, x_n)= d\nu(x_0)\,  d\mathcal{B}_{x_0}(x_1)\, \cdots \, d\mathcal{B}_{x_{n-1}}(x_n).$$

Given a stationary chain segment $X_0, \dots, X_n$, let $(Y_0, \dots, Y_n)=\mathcal{R}(X_0, \dots, X_n)$ be its proper time reversal. We introduce the operator $\mathcal{B}^-$ defined, for now, on bounded continuous functions:
$$ (\mathcal{B}^-f)(x):= \mathbb{E}\left[f(Y_{i+1})|Y_i=x\right].$$

  We previously defined the probability measure $\mathbb{P}_{[0,n]}$ on the space of finite chain segments $\mathcal{D}_{[0,n]}$. 
Correspondingly, we define the probability measure $\mathbb{P}^-_{[0,n]}$ on the same space by
$$d\mathbb{P}^-_{[0,n]}(x_0, \dots, x_n)= d\nu(x_0)\,  d\mathcal{B}^-_{x_0}(x_1)\, \cdots \, d\mathcal{B}^-_{x_{n-1}}(x_n).$$

Following \cite{qian} we make the   definitions given next.
\begin{definition}[Relative entropy]
Suppose that $\mathbb{P}_1$ and $\mathbb{P}_2$ are two probability measures on a measurable space $(\mathcal{D},\mathcal{F})$. The {\em relative entropy} of $\mathbb{P}_1$ with respect to $\mathbb{P}_2$ is defined as
$$ H(\mathbb{P}_1, \mathbb{P}_2):=\begin{cases}
\int_{\mathcal{D}} \log\left( \frac{d\mathbb{P}_1}{d\mathbb{P}_2}\right)\, d\mathbb{P}_1 &\text{ if } \mathbb{P}_1\ll\mathbb{P}_2 \text{ and } \log\left( \frac{d\mathbb{P}_1}{d\mathbb{P}_2}\right) \in L^1(\mathcal{D}, \mathbb{P}_1)\\
+\infty &\text{ otherwise.}
\end{cases}$$
\end{definition}
\begin{definition}[Entropy production rate]
The {\em entropy production rate} of the stationary Markov chain $X_0, X_1, \dots$ defined by $\nu$ and $\mathcal{B}$ is defined by
$$ e_p:= \lim_{n\rightarrow \infty}\frac1n H\left(\mathbb{P}_{[0,n]},\mathbb{P}^-_{[0,n]}\right)$$
where $H\left(\mathbb{P}_{[0,n]},\mathbb{P}^-_{[0,n]}\right)$ is the relative entropy of $\mathbb{P}_{[0,n]}$ with respect to
$\mathbb{P}^-_{[0,n]}$ restricted to the $\sigma$-algebra generated by $X_0, \dots, X_n$. 
\end{definition}

Let $\nu_q$ be the probability measure on $N^+_q$ obtained  by disintegrating  $\nu$ with respect to $\pi:N^+\rightarrow \partial M$. So, if $x=(q,u)\in N^+$,
    $$ d\nu(x) = d\nu_q(u)\, d(\pi_*\nu)(q).$$
    Recall that $\nu^-=\mathcal{T}_*\nu$, and $\mu_q^+$ is the Maxwellian at boundary point $q$ for temperature $T(q)$. 
    The following assumptions  will be made concerning the stationary measure $\nu$:
    \begin{enumerate}
    \item  $\pi_*\nu\ll A$ where $A$ is the $(n-1)$-dimensional Riemannian volume on $\partial M$;
    \item  The measures $\nu_q$, $\mu^+_q$, $J_*\nu_q^-$ are mutually equivalent. 
    \end{enumerate}
   
 We  define the measures $\mu^\pm\in \mathcal{P}(N^\pm)$ by 
   $$d\mu^\pm(x)=d\mu_q^\pm(u)\, d\bar{A}(q)$$ 
   where $\bar{A}=A/A(\partial M)$ is the normalized Riemannian volume measure on $\partial M$.  Moreover, let $S \subset N $ denote the bundle of unit vectors in $N$, $S^+=S \cap N^+$, and $S^+_q=S^+\cap N_q$.  Let $\sigma \in \mathcal P(S^+)$ be given by
\begin{equation}\label{billiard measure}
d\sigma(q,u)= C\left\langle u, \mathbbm{n}_q\right\rangle\,  dV_q^{\text{\tiny S}}(u)\, dA(q),
\end{equation}
where $C$ is a normalizing constant and $V_q^{\text{\tiny S}}$ is the Riemannian volume measure on $S_q^+$.

    \begin{theorem}[Second law of thermodynamics]\label{heat entropy}
     Let  $E_0(q,u)=\frac12 m|u|^2_q$ denote the kinetic energy function and let $\Phi(q)$ be a bounded measurable function defined on the boundary of $M$ which denotes the potential function of the system.
 Suppose that a stationary probability measure $\nu$ for the random billiard map $\mathcal{B}$ exists and satisfies assumptions 1 and 2 above.
 Then
 \begin{equation*} e_p= - \bigintssss_{\partial M}\frac{\nu^+_q(E_0)-\nu^-_q(E_0)}{\kappa T(q)}\, d\bar{A}(q) - \bigintssss_{N^+}\frac{\Phi(\pi(\mathcal{T}(x))) - \Phi(\pi(x))}{\kappa T(\pi(x))}
 \, d\nu(x)\geq 0.\end{equation*}
 In words, the entropy production rate (per collision with the boundary) is the average of the energy gained at each iteration of the random billiard system   divided by the temperature at the point of collision.
 \end{theorem}
 
 That the displayed quantity in the theorem is non-negative means that, in the absence of a potential energy function, 
 the forward direction in time for the Markov chain is distinguished from the time-reversed process in the following sense: on average,  energy is extracted from the regions of  higher temperature   of $\partial M$ and released  into the regions of lower temperature. 
 
With this theorem in hand, the importance of systematically understanding the stationary measure $\nu$ for random billiard systems is readily seen.  The rest of our main results are centered around the study of the stationary measure for a class of examples whose reflection operator is of the Maxwell-Smoluchowski kind.

\begin{definition}[Maxwell-Smoluchowski model]\label{def:ms-model}
Let $\text{\em Ref}_q$ denote the specular reflection at the regular boundary point $q$, and fix 
 $\alpha(q)\in (0,1]$.  
 Let $\mu_q^+$ be the Mawellian at $q$.
 Define   $P_q$ by its evaluation on a test function $f$ as follows:
 $$ P_x(f)=\alpha(q)\mu_q^+(f)+(1-\alpha(q)) f(\text{\em Ref}_qx).$$
 Thus the  surface scattering process defined by this  general reflection operator, known as the {\em Maxwell-Smoluchowski model},  
 has the effect of mapping a pre-collision velocity $v$ of an incoming particle at $q$ to the random post-collision velocity $V$
 whose probability distribution is $\mu_q^+$ with probability $\alpha(q)$ and the specular reflection of $v$ with probability $1-\alpha(q)$.  
\end{definition}

The next theorem, informally stated below and stated precisely in Theorem \ref{uniform ergodicity} and Proposition \ref{entropy formula proposition} of Section \ref{sec:Multi-temperature Maxwell-Smoluchowski systems}, considers a random billiard system with Maxwell-Smoluchowski reflection operator whose boundary $\partial M$ is partitioned into $N$ components $\Gamma_i$ with temperatures $T(q) \equiv T_i$ and constants $\alpha(q) \equiv \alpha_i$, for $i = 1,\ldots, N$.

 \begin{theorem}
 Under mild regularity conditions on $\partial M$ (see Assumption \ref{Lipschitz condition}), the random billiard Markov chain $(X_n)_{n\geq 0}$ is uniformly ergodic.  Moreover, when the boundary temperatures $T_1,\ldots,T_N$ are equal, say to a constant $T_0 > 0$, the stationary distribution $\nu$ is given by $$d\nu(q,v) = \rho_q(v)dV_q(v)d\bar A(q)$$ where $\rho_q(v)$ is the Maxwellian density given in (\ref{maxwellian}) with constant temperature $T(q) \equiv T_0$.  When the boundary temperatures are not equal, the stationary distribution, expressed in polar coordinates as a measure on $S^+ \times (0,\infty)$ is given by $d\nu(q,u,r) = d\nu^{\mathrm s}_q(r)d\sigma(q,u)$, where $\nu^\mathrm{s}_q$ is the stationary distribution of the speed after collision with boundary point $q$ and $\sigma$ is defined in Equation (\ref{billiard measure}).  The measure $\nu_q^\mathrm s$ is constant on components $\Gamma_i \subset \partial M$, so we let $\nu_i^\mathrm s := \nu_q^\mathrm s$ for $q \in \Gamma_i$.  Letting $\nu^\mathrm s$ be the $N$-dimensional vector with components $\nu_i^\mathrm s$, $$\nu^\mathrm s = (I - Q)^{-1}\pi,$$ where
$Q$ is an $N \times N$ matrix and $\pi$ an $N$-dimensional vector given by $$Q_{ij} = (1-\alpha_i)p_{ij} A_i/A_j, \quad \pi_i = \bar A_i\alpha_i\mu_i^\mathrm s$$ with $A_i = A(\Gamma_i)$ and $\bar A_i = \bar A(\Gamma_i)$ the volume measure and normalized volume measure of boundary component $\Gamma_i$ respectively.  Finally, 
 the entropy production rate is given by
$$
e_p = -\sum_{j =1}^N \frac{\nu_j^+(E_0) - \nu_j^-(E_0)}{\kappa T_j} = -\sum_{i,j =1}^N \frac{\nu_j^+(E_0) - \nu_i^+(E_0)}{\kappa T_j} p_{ij}\frac{A_i}{A_j}
$$
where $\nu^+ = \nu$, $\nu^- = \mathcal T_* \nu$, $\nu_j^\pm$ is the restriction to $N_j^\pm = \{(q,u) \in N^\pm : q \in \Gamma_i\}$, and $p_{ij}$ is the conditional probability of the billiard particle colliding next with boundary component $\Gamma_j$ given that its current position is on boundary component $\Gamma_i$.  

 \end{theorem}

In Subsection \ref{subsec:Examples}, these results are used to compute the entropy production rate explicitly for a series of examples, emphasizing the influence of system parameters.

\section{Basic facts about the Markov chain}\label{sec:Basic facts}

Consider the Hilbert space $\mathcal{H}:=L^2(N^+,\nu)$.  No confusion should arise if we use the same notation $\langle\cdot, \cdot\rangle$ for the inner product on $\mathcal{H}$  as for the Riemannian metric on $M$. We now consider   $\mathcal{B}$  as   an operator on $\mathcal{H}$. 
 \begin{proposition}
 The billiard map $\mathcal{B}$, regarded as an operator on $\mathcal{H}$, has norm $\|\mathcal{B}\|_2=1$. 
 \end{proposition}
 \begin{proof}
This   is an  immediate consequence of the Cauchy-Schwarz inequality and the fact that $\nu$ and $P_x$ are probability measures:
$$\|\mathcal{B}f\|^2_2=\int_{N^+}|P_{\mathcal{T}(x)}(f)|^2\, d\nu(x)\leq \int_{N^+}\int_{N^+_{\pi(\mathcal{T}(x))}} |f(y)|^2\, dP_{\mathcal{T}(x)}(y)\, d\nu(x)=(\nu\mathcal{B})(|f|^2).$$
But stationarity implies $(\nu\mathcal{B})(|f|^2)=\nu(|f|^2)=\|f\|_2^2$, so that $\|\mathcal{B}\|_2\leq 1$. The observation $\mathcal{B}1=1$ concludes the proof.
 \end{proof}

 \begin{proposition}[Transition operator for the time-reversed chain]\label{Bminus}
Suppose that $\nu$ and $\mathcal{J}_*\nu$ are equivalent measures, so that the Radon-Nikodym derivative $\rho^\nu:=\frac{d(\mathcal{J}_*\nu)}{d\nu}$ is defined.  Then
$ \mathcal{B}^- = \rho^\nu \mathcal{J} \mathcal{B}^* \mathcal{J}.$
Here $\mathcal{B}^*$ is the adjoint operator to $\mathcal{B}$, $\mathcal{J}$ is the composition operator $\mathcal{J}f:=f\circ \mathcal{J}$, and $\rho^\nu$ is   identified with its multiplication operator.
 \end{proposition}
  \begin{proof}
Let $X$ be a random variable with probability distribution $\nu$. Let the function $f$ on $N^+$ be in the domain of $\mathcal{B}^-$ and consider, for $g\in \mathcal{H}$, the inner product
$$ \langle\mathcal{B}^-f, g\rangle= \int_{N^+}(\mathcal{B}^- f)(x)\overline{g(x)}\, d\nu(x) =\mathbb{E}\left[(\mathcal{B}^- f)(X)\overline{g(X)}\right]=\mathbb{E}\left[\mathbb{E}\left[f(Y_{i+1})|Y_i=X\right]\overline{g(X)}\right].$$
Now $Y_i= \mathcal{J}X_{n-i}$ and $Y_{i+1}= \mathcal{J}X_{n-i-1}$. Letting $l=n-i-1$, the rightmost term above becomes
$$\mathbb{E}\left[\left.\mathbb{E}\left[\overline{(g\circ \mathcal{J})(X_{l+1}))}(f\circ \mathcal{J})(X_l)\right| \mathcal{J}X_{l+1}=X \right]\right] =
\mathbb{E}\left[\overline{(g\circ \mathcal{J})(X_{l+1}))}(f\circ \mathcal{J})(X_l)\right].$$
We then  have
\begin{align*}\mathbb{E}\left[\overline{(g\circ \mathcal{J})(X_{l+1}))}(f\circ \mathcal{J})(X_l)\right]&=\mathbb{E}\left[\mathbb{E}\left[\left. \overline{(g\circ\mathcal{J})(X_{l+1})}\right|X_l=X\right] (f\circ \mathcal{J})(X)\right]\\
&=\mathbb{E}\left[\overline{\mathcal{B}_X(g\circ\mathcal{J})} (f\circ \mathcal{J})(X)\right]\\
&=\left\langle \mathcal{J}f, \mathcal{B}\mathcal{J} g\right\rangle.
\end{align*}
Therefore $\mathcal{B}^-= \mathcal{J}^*\mathcal{B}^*\mathcal{J}$.  Since
$\left\langle  \mathcal{J}f,g\right\rangle= \int_{N^+} (f\circ \mathcal{J})\overline{g}\, d\nu = \int_{N^+} f\overline{(g\circ \mathcal{J})}\rho^\nu\, d\nu=\left\langle f, \rho^\nu \mathcal{J} g\right\rangle$ we obtain $\mathcal{J}^*=\rho^\nu \mathcal{J}$, concluding the proof.
  \end{proof}
A similar argument shows that the process corresponding to the simple reversal $$(X_0, \dots, X_n)\mapsto (X_n, \dots, X_0)$$ has transition operator $\mathcal{B}^*$. Moreover $\nu\mathcal{B}^*=\nu$ since $\nu\mathcal{B}^*f= \langle\mathcal{B}^*f, 1\rangle=\langle f, \mathcal{B}1\rangle = \langle f, 1\rangle = \nu(f).$ On the other hand,  $\nu\mathcal{B}^-=\mathcal{J}_*\nu$ and it is not the case in general  that $\nu=\mathcal{J}_*\nu$, so $\nu$ may not be stationary with respect to $\mathcal{B}^-$.

 From the billiard map $\mathcal{B}$ and a stationary probability measure $\nu$ we define the probability measure $\eta\in \mathcal{P}(\mathcal{D})$ by $d\eta(x,y)=d\nu(x)\, d\mathcal{B}_x(y)$ and call $\eta$ the probability measure on {\em forward pairs}. The probability measure on {\em backward pairs} $\eta^-\in \mathcal{P}(\mathcal{D})$ is defined by $\eta^-=\mathcal{R}_*\eta$, where we recall that $\mathcal{R}(x,y)=(\mathcal{J}y,\mathcal{J}x)$. We assume that $\eta$ and $\eta^-$ are in the same measure class. 
  
  \begin{proposition} The measure $\eta^-$  satisfies
  $d\eta^-(x,y)=d\nu(x)\, d\mathcal{B}^-_x(y)$ for $(x,y)\in \mathcal{D}$. 
  \end{proposition} 
  \begin{proof}
It suffices to show that the two measures, $\eta^-$ and the measure defined by the right-hand side of the equation, give the same integral on  functions of the form $f\times g: (x,y)\mapsto f(x)g(y)$ where $f$ and $g$ are bounded continuous functions. In fact,
$$\eta^-(f\times g)= (\mathcal{R}_*\eta)(f\times g)= \eta((f\times g)\circ \mathcal{R})=\eta((g\circ \mathcal{J})\times(f\circ \mathcal{J})).$$
The right-most term above is equal to
$$ \int_{\mathcal{D}} g(\mathcal{J}x)f(\mathcal{J}y) d\mathcal{B}_x(y)\, d\nu(x)=\left\langle \mathcal{J}g, \overline{\mathcal{B}\mathcal{J}f}\right\rangle=\left\langle \mathcal{J}^*\mathcal{B}^*\mathcal{J}g, \overline{f}\right\rangle=\left\langle\mathcal{B}^-g,\overline{f}\right\rangle.$$
Finally, the equality $\left\langle\mathcal{B}^-g,\overline{f}\right\rangle=\int_{\mathcal{D}}(f\times g)(x,y)\, d\mathcal{B}^-_x(y)\, d\nu(x)$
 concludes the proof. 
  \end{proof}

 Note that $\mathbb{P}_{[0,n]}$ and $\mathbb{P}^-_{[0,n]}$  are equivalent under the assumption that $\eta$ and $\eta^-$ are equivalent. 
            
 Observe that 
 $$ \frac{d\mathbb{P}^-_{[0,n]}}{d\mathbb{P}_{[0,n]}}(x_0, \dots, x_n)=\frac{d\mathcal{B}^-_{x_0}}{d\mathcal{B}_{x_0}}(x_1)\cdots
 \frac{d\mathcal{B}^-_{x_{n-1}}}{d\mathcal{B}_{x_{n-1}}}(x_n)=\frac{d\eta^-}{d\eta\ }(x_0,x_1)\cdots \frac{d\eta^-}{d\eta\ }(x_{n-1},x_n).$$
 Also note that, for any function $f$ on $\mathcal{D}$
 $$\int_{\mathcal{D}_{[0,n]}} f(x_i,x_{i+1})\, d\mathcal{B}_{x_{n-1}}(x_{n}) \cdots d\mathcal{B}_{x_0}(x_1)\, d\nu(x_0) =
 \int_{\mathcal{D}}f(x_i,x_{i+1})\, d\mathcal{B}_{x_i}(x_{i+1})\, d\nu(x_i)=\int_{\mathcal{D}} f\, d\eta.$$ 
 From these observations we immediately   obtain
\begin{equation}\label{H}H\left(\mathbb{P}_{[0,n]},\mathbb{P}^-_{[0,n]}\right)=-n\int_{\mathcal{D}}\log\left(\frac{d\eta^-}{d\eta}\right)\, d\eta.\end{equation}
   \begin{proposition}[Entropy production rate]\label{entropy_general}
   The entropy production rate for the random billiard system, under the assumption that the probabilities on pairs $\eta$ and $\eta^-$ are equivalent,  takes the form
   $$ e_p=\frac12\int_{\mathcal{D}} \left[d\eta -d\eta^-\right]\log\left(\frac{d\eta\ }{d\eta^-}\right).$$
   In particular, this expression shows that $e_p\geq 0$. 
   \end{proposition}   
      \begin{proof}
      Due to Equation (\ref{H})   we have $e_p=-\int_{\mathcal{D}}\log\left(\frac{d\eta^- }{d\eta\ }\right)\, d\eta$.  Now observe that
      $\frac{d\eta^-}{d\eta\ }\circ\mathcal{R}=\frac{d\eta\ }{d\eta^-}.$ In fact, for any measurable set $E\subset\mathcal{D}$,
      $$ \int_E \frac{d\eta^-}{d\eta\ }\circ\mathcal{R}\, d\eta^- = \int_{\mathcal{R}(E)} \frac{d\eta^-}{d\eta\ }\, d(\mathcal{R}_*\eta^-) =\int_{\mathcal{R}(E)}\, d\eta^-=\eta(E)=\int_E \frac{d\eta\ }{d\eta^-}\, d\eta^-.$$
      Therefore,
      $$e_p=-\int_{\mathcal{D}}\log\left(\frac{d\eta^- }{d\eta\ }\right)\, d\eta=  -\int_{\mathcal{D}} \log\left(\frac{d\eta\ }{d\eta^-}\right)\, d\eta^-$$
      from which we conclude that $e_p=\frac12\int_{\mathcal{D}} \left[d\eta -d\eta^-\right]\log\left(\frac{d\eta\ }{d\eta^-}\right)$ as claimed. It is apparent from this expression that $e_p\geq 0$.
      \end{proof}

\section{Second Law of Thermodynamics}\label{sec:Second Law of Thermodynamics}

    The reciprocity property imposed on the reflection operator $P$, which is needed 
    in order to make sense of  the concept of boundary temperature,
     has not been used so far. Below, we rewrite  the expression for $e_p$ obtained in the previous section,  making use of this property.  Before doing 
     so,      it  is useful to 
introduce  a generalized but natural notion of reciprocity as noted after   Definition \ref{reciprocity}.     
  This yields a more general  notion of  reflection operator  that applies to manifolds having a local product structure, corresponding to billiard systems consisting of multiple rigid masses. In such cases we suppose the existence of a measurable family of subspaces $W_q\subset T_qM$ for each $q\in \partial M$ such that 
$\mathbbm{n}_q\in W_q$,  and define the Maxwellian  $\mu_q^\pm$ as in Equation (\ref{maxwellian}), except that the dimension $n$ is now  replaced with the dimension of $W_q$ and $dV_q$ is replaced with the volume measure on $W_q$.   We then assume that the family of operators $P_q$ satisfies
\begin{enumerate}
\item $u\in W_q^-\mapsto \mathcal{P}(W_q^+)$, where $W_q^\pm :=W_q\cap N^\pm_q$;
\item if $u\in T_q(\partial M)$ is perpendicular to $W_q$, then $P_{(q,u)}$ is the point mass at $u$;
\item reciprocity is defined for the measure $d\zeta_q(u,v)$ for $(u,v)\in W_q^-\times W_q^+$. 
\end{enumerate}
\begin{definition}\label{thermal direction}
The subspaces $W_q$ will be called {\em directions of thermal contact} or simply {\em thermal directions}. The boundary of $M$ is said to have temperature $T$ at $q\in \partial M$ if 
$P_q$ satisfies reciprocity with respect to the Maxwellian on $W_q^-$ having parameter $\beta=1/\kappa T$.
\end{definition}

\begin{wrapfigure}{R}{0.4\textwidth}
\begin{center}
\includegraphics[width=0.35\textwidth]{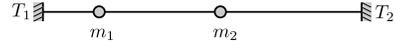}\ \ 
\caption{{\small The two-masses system. }}
\label{twomasses}
\end{center}
\end{wrapfigure}

Thus  if $u=u_1+u_2\in N_q^-$ is a pre-collision velocity decomposed into $u_1\in W_q^-$ and $u_2$ in the orthogonal complement   $W_q^\perp$ of $W_q$, then
the post-collision velocity is $U_1+u_2$ where $U_1$ is a random vector in $W^+_q$ distributed according to $P_{(q,u_1)}$, and reciprocity holds with respect to a Maxwellian on the space of thermal directions.

A   simple example will help to clarify the need for the above notion of thermal directions. Consider the system shown in  Figure \ref{twomasses}
describing  two point masses $m_1, m_2$ that can slide freely over an interval of length $l$.  
When the two masses collide with each other, their post-collision velocities are derived from the assumptions of conservation of kinetic energy and momentum 
and when they collide with the end-points of the interval, they reflect according to random reflection operators at temperatures $T_1$ and $T_2$.

The configuration space of the pair of masses is a right-triangle with the sides adjacent to the right angle having length $l$. A point $(x,y)$ represents the configuration in which $m_1$ is at $x$ and $m_2$ is at $y$. On the  longer side are the  configurations representing collisions of the two masses. Introducing new coordinates
$x_1=\sqrt{\frac{m_1}{m}}x, \ \ \ x_2=\sqrt{\frac{m_2}{m}}y,$
where $m=m_1+m_2$,
the total kinetic energy  becomes a multiple of the square Euclidian norm,  $|v|^2$ of the velocity vector $v=(\dot{x}_1, \dot{x}_2)$.

In this  rescaled picture, the two-masses system becomes a random  billiard system    in which  
a point particle of mass $m=m_1+m_2$   moves freely inside the triangle and undergoes specular reflection on the hypotenuse while reflection on the shorter sides is random.  On these sides the thermal direction is along the normal vector $\mathbbm{n}_q$ and their temperature is $T_i$, $i=1,2$.

\begin{wrapfigure}{L}{0.4\textwidth}
\begin{center}
\includegraphics[width=0.35\textwidth]{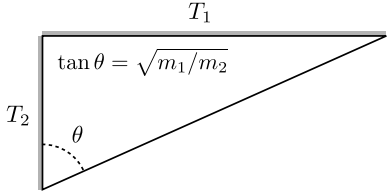}\ \ 
\caption{{\small Configuration manifold for the two-masses random billiard system. }}
\label{triangle_twomasses}
\end{center}
\end{wrapfigure}

As this example makes clear (see also the idealized heat engine given in Section \ref{sec:Efficiency and work production}),  the moving particle in our definition of  random billiards can be thought in general to represent the configuration 
of a system consisting of several moving rigid masses, possibly extended rigid bodies in dimension $3$. In such cases the configuration manifold can be a non-flat  Riemannian manifold.

In preparation for the proof of Theorem \ref{heat entropy} let us recall  that the measures $\mu^\pm\in \mathcal{P}(N^\pm)$ were defined as 
   $$d\mu^\pm(x)=d\mu_q^\pm(u)\, d\bar{A}(q)$$
   where $\bar{A}=A/A(\partial M)$ is the normalized Riemannian volume measure on $\partial M$. 
    When  the space $W_q$ of thermal directions is not all of $N_q$, we let  $d\mu_q^\pm(u)=d\overline{\mu}_q(u_1)\, dS^{\perp}_q(u_2)$. This is the product measure of the Maxwellian along $W^\pm_q$ and the normalized volume measure on the 
    hemisphere of radius $|u_2|_q$, where $u=u_1+u_2$ is the orthogonal decomposition of $u$ into its $W_q$ and $W_q^\perp$ components.

   \begin{proposition}\label{entropy_reciprocity}
   With the  definitions from Section \ref{sec:Basic facts}, and bringing into play the reciprocity property of  the reflection operator $P$, we obtain
\begin{equation}\label{prop_equation} \frac{d\eta\ }{d\eta^-}(x,y)=\frac{d\nu}{d(\mathcal{J}_*\mu^+)}(x)\left[\frac{d\nu}{d(\mathcal{J}_*\mu^+)}(\mathcal{J}y)\right]^{-1}\end{equation}
    where $x=(q,u)$.
   \end{proposition}
  \begin{proof}
First observe that $\frac{d(\mathcal{J}_*\nu)}{d\mu^+}\circ \mathcal{J}=\frac{d\nu}{d(\mathcal{J}_*\mu^+)}$. In fact, for any measurable subset $E\subset N^+$,
$$\int_E\frac{d(\mathcal{J}_*\nu)}{d\mu^+}(\mathcal{J}x)\, d(\mathcal{J}_*\mu^+)(x)=\int_{\mathcal{J}(E)}\frac{d(\mathcal{J}_*\nu)}{d\mu^+}(x)\, d\mu^+(x)=\int_{\mathcal{J}(E)}d(\mathcal{J}_*\nu)=\int_E d\nu.$$
The identity then follows from  $\int_E d\nu=\int_{N^+} \frac{d\nu}{d(\mathcal{J}_*\mu^+)}(x)\, d(\mathcal{J}_*\mu^+)(x).$

Proceeding with the proof of  Equation (\ref{prop_equation}), we first recall that
$$
d\eta(x,y)= d\nu(x)\, d\mathcal{B}_x(y)=d\nu(x)\, dP_{\mathcal{T}(x)}(y)\ \text{ and }\
d\eta^-(x,y)=d\eta(\mathcal{J}y,\mathcal{J}x)=d\nu(\mathcal{J}y)\, dP_{Jy}(\mathcal{J}x).
$$
Reciprocity was defined by the relation
$$ d\mu^-_q(Jy)\, dP_{Jy}(\mathcal{J}x) = d\mu^-_{q'}(\mathcal{T}(x))\, dP_{\mathcal{T}(x)}(y).$$
(See Figure \ref{billiard_definitions}.) 
Using  the measures $\mu^\pm$  on $N^\pm$ defined just prior to the statement of this proposition, we may rewrite the reciprocity property as
$$ d\mu^-(Jy)\, dP_{Jy}(\mathcal{J}x) = d\mu^-(\mathcal{T}(x))\, dP_{\mathcal{T}(x)}(y).$$

\begin{figure}[htbp]
\begin{center}
\includegraphics[width=3.0in]{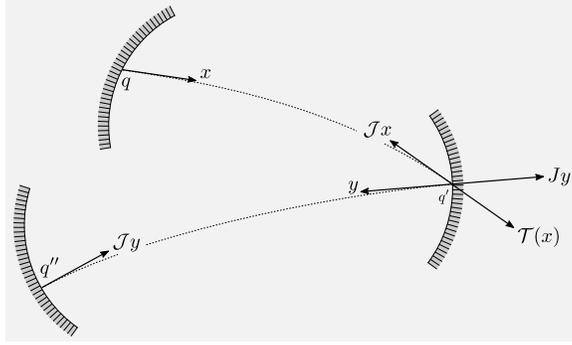}\ \ 
\caption{{\small Depiction of some of the vectors appearing in the proof of Proposition \ref{entropy_reciprocity}.  }}
\label{billiard_definitions}
\end{center}
\end{figure}  

Also notice that $d\mu^\pm(Jy)=d(J_*\mu^\pm)(y)=d\mu^\mp(y)$ and $d\mu^-(\mathcal{T}(x))= d\mu^+(\mathcal{J}x)$. 
Therefore, 
\begin{align*}
d\eta^-(x,y)&=\frac{d(\mathcal{J}_*\nu)}{d\mu^+}(y)\, d\mu^+(y)\, dP_{Jy}(\mathcal{J}x)\\
&=\frac{d(\mathcal{J}_*\nu)}{d\mu^+}(y)\, d\mu^+(\mathcal{J}x)\, dP_{\mathcal{T}(x)}(y)\\
&=\frac{d(\mathcal{J}_*\nu)}{d\mu^+}(y)\, \frac{d(\mathcal{J}_*\mu^+)}{d\nu}(x) \, d\nu(x)\, dP_{\mathcal{T}(x)}(y)\\
&= \frac{d(\mathcal{J}_*\nu)}{d\mu^+}(y)\, \frac{d(\mathcal{J}_*\mu^+)}{d\nu}(x) \, d\eta(x,y).
\end{align*}
This, in combination with the observation that began the proof, yields Equation (\ref{prop_equation}).
  \end{proof}
  
The factorization of the Radon-Nikodym derivative $d\eta/d\eta^-$ as a product of a function of $x$ and a function of $y$, as given in Proposition \ref{entropy_reciprocity},      
allows us to express $e_p$ as an integral over $N^+$ rather than $\mathcal{D}$. This is indicated in the following proposition.

\begin{proposition}\label{variant ep}
Define the function $\Lambda:N^+\rightarrow (0,\infty)$ given by $\Lambda(x)=\frac{d\nu}{d(\mathcal{J}_*\mu^+)}(x)$. Then, given a stationary probability measure $\nu$ of the random billiard system,
$$e_p=\frac12\nu\left(\log\frac{\Lambda}{\Lambda\circ\mathcal{J}}\right)$$
under the assumption   that $\eta$ and $\eta^-$ are equivalent measures.
\end{proposition}   
 \begin{proof} 
 Observe that for any $\nu$-integrable function $f$ on $N^+$,
\begin{equation}\label{simplify}\int_{\mathcal{D}} f(x) \, d\eta^\pm(x,y)= \int_{N^+} \left[\int_{N^+_{\pi(\mathcal{T}(x))}} d\mathcal{B}^\pm_x(y)\right] f(x)\, d\nu(x)=\nu(f). \end{equation}  
Using  Proposition \ref{entropy_reciprocity} and  the general expression for the entropy production rate given in Proposition \ref{entropy_general} we obtain
 \begin{align*}
 e_p&=\frac12\int_{\mathcal{D}}\left[d\eta(x,y)-d\eta^-(x,y)\right] \log\left(\frac{\Lambda(x)}{\Lambda(\mathcal{J}y)}\right)\\
 &=-\frac12\int_{\mathcal{D}}\left[d\eta(x,y)-d\eta^-(x,y)\right] \log\left(\Lambda(\mathcal{J}y)\right)
 \end{align*}
 where the term involving $\Lambda(x)$ vanished due to Equation (\ref{simplify}). Also observe 
 that 
 $$\int_{\mathcal{D}}\log\left(\Lambda(\mathcal{J}y)\right)\, d\eta^-(x,y)=\int_{\mathcal{D}} \log\left(\Lambda(x)\right)\, d\eta(x,y)=\nu\left(\log \Lambda\right)$$
and  that
$$\int_\mathcal{D}\log\left(\Lambda(\mathcal{J}y)\right)\, d\eta(x,y)=\int_{N^+} \left[\int_{N^+_{\pi(\mathcal{T}(x))}} \log\left(\Lambda(\mathcal{J}y)\right)\,
d\mathcal{B}_x(y)\right]\, d\nu(x)=\nu\mathcal{B}(\log\Lambda\circ\mathcal{J})=\nu(\log\Lambda\circ\mathcal{J}).$$
 Collecting the terms we obtain $e_p=\frac12 \nu\left(\log\left(\Lambda/\Lambda\circ \mathcal{J}\right)\right)$ as claimed.
 \end{proof}

    Let $m$ be the measure on $N^+$ defined by $dm(q,u)=dV_q(u)\, d\bar{A}(q)$ and $g(q,u)$ the density of $\nu$ with respect to $m$.  Notice that

    $$ d\mu^+(q,u)=C_q\exp\left\{-\beta(q)E(q,u)\right\}\, dm(q,u)$$
  where $E$ is the sum of the kinetic and potential energy functions: $E(q,u) = E_0(q,u)+\Phi(q)$, $E_0(q,u)=\frac12 m|u|^2_q$. 
  Furthermore, invariance of $m$ under $\mathcal{J}$ gives
  $$d(\mathcal{J}_*\mu^+)(x) = \frac{h(\mathcal{J}x)}{h(x)}\, d\mu^+(x) = h(\mathcal{J}x)\, dm(x)$$
  where $h(x)=h(q,u)=C_q\exp\left\{-\beta(q) E(q,u)\right\}$.  Define 
  $ l(x,y)=g(x)h(\mathcal{J}y)$
  for $(x,y)\in \mathcal{D}$. These definitions  give the expression
\begin{equation}\label{ratio eta}\frac{d\eta\ }{d\eta^-}(x,y) = \frac{\Lambda(x)}{\Lambda(\mathcal{J}y)} = \frac{g(x)}{h(\mathcal{J}x)}\frac{h(\mathcal{J}y)}{g(y)} = \frac{l(x,y)}{(l\circ \mathcal{R})(x,y)}, \end{equation}
where $l(x,y)=g(x) h(\mathcal{J}y)$.
 These definitions will be used in the proof of the Second Law. 

 \begin{proof}[Proof of Theorem \ref{heat entropy}]
The expression in the theorem statement could be derived by taking as starting point the expression for $e_p$ given in Proposition \ref{variant ep}, but we prefer to 
begin with the general form for $e_p$ asserted in Proposition \ref{entropy_general} and the function $l(x,y)$ appearing in (\ref{ratio eta}). In fact, we start with the non-symmetrized expression for $e_p$ and use (in the fifth line) the property expressed in Equation (\ref{simplify}):
\begin{align*}
e_p &= \int_{\mathcal{D}}d\eta(x,y)\log \frac{d\eta}{d\eta^-}(x,y)\\
&= \int_{\mathcal{D}}d\eta(x,y)\log \frac{l(x,y)}{l(\mathcal{R}(x,y))}\\
&=\int_{\mathcal{D}}\left[d\eta(x,y)-d\eta^-(x,y)\right] \log l(x,y)\\
&=\int_{\mathcal{D}}\left[d\eta(x,y)-d\eta^-(x,y)\right] \log g(x) - \int_{\mathcal{D}}\left[d\eta(x,y)-d\eta^-(x,y)\right] \log h(\mathcal{J}y)\\
&=- \int_{\mathcal{D}}\left[d\eta(x,y)-d\eta^-(x,y)\right] \log h(\mathcal{J}y).
\end{align*}
Equation (\ref{simplify}), again, implies
$$ \int_{\mathcal{D}} d\eta^-(x,y)\log h(\mathcal{J}y)= \int_{\mathcal{D}} d\eta(x,y) \log(h(x))= \nu(\log(h)).$$
And stationarity of $\nu$ implies
$$  \int_{\mathcal{D}} d\eta(x,y)\log h(\mathcal{J}y)=\int_{N^+}d\nu(x) \mathcal{B}_x(\log(h\circ\mathcal{J})) = \nu\mathcal{B}(\log(h\circ \mathcal{J})) = \nu(\log(h\circ \mathcal{J})).$$
Therefore
$$ e_p=\int_{N^+} \log\left[\frac{h(x)}{h(\mathcal{J}(x))}\right]\, d\nu(x).$$
Now observe that  $\mathcal{J}_*\nu = J_*\mathcal{T}_*\nu = J_*\nu^-$, which allows us to write
$$ e_p=\int_{N^+}\left[d\nu^+(x)-d(J_*\nu^-)(x)\right] \log(h(x)).$$
Note that the kinetic energy function $E_0$ is invariant under $J$ and that $$\log h(x) = \log C_q - \beta(q) E(x).$$ The integral of the
constant term $\log C_q$ against the difference of probability measures $\nu^+_q- J_*\nu^-_q$ gives $0$, so we are left with
$$e_p =-  \int_{N^+}\left[d\nu^+(x) - d(J_*\nu^-)(x)\right](\beta E)(x).  $$
Decomposing along the fibers of $\pi:N^+\rightarrow \partial M$ gives the desired expression.
 \end{proof}

\section{Multi-temperature Maxwell-Smoluchowski systems}\label{sec:Multi-temperature Maxwell-Smoluchowski systems}

The central purpose of this section is to study the entropy production rate for a system whose configuration space $M$ is a Riemannian manifold, with some mild regularity conditions, whose boundary is partitioned into components $\Gamma_i$ which are kept at temperatures $T_i$ respectively for $i = 1,\ldots, N$.  The notion of a thermostatted boundary is modeled using the Maxwell-Smoluchowski model, defined Definition \ref{def:ms-model}.  Recall that the model can be thought of as follows.  Let constants $\alpha_1,\ldots,\alpha_N \in (0,1]$ be given.  Upon collision at any point of boundary component $i$, the post-collision velocity of the colliding particle is either chosen randomly according to the Maxwell-Boltzmann distribution with temperature $T_i$, with independent probability $\alpha_i$, or the particle reflects specularly with probability $1-\alpha_i$.  When $\alpha_i$ is small, this model may be regarded as a random perturbation of an ordinary billiard system.  More generally, it can be thought of as a model of thermalization, where the particle only takes on the temperature of the boundary thermostat after a geometrically distributed number of collisions.

The proof that the Maxwell-Smoluchowski model satisfies reciprocity amounts to the following elementary exercise.
(In order to alleviate clutter, we omit 
the subscript $q$ from maps and measures.)
 For proving invariance of  $\zeta$ under $\mathcal{R}$ it is sufficient to test the equality $\mathcal{R}_*\zeta=\zeta$ on   functions of the form $(f\times g)(u,v)=f(u)g(v)$.
 For such functions, 
\begin{align*}\zeta(f\times g)-\left(\mathcal{R}_*\zeta\right)(f\times g)&= \zeta\left(f\times g\right) - \zeta\left((f\times g)\circ\mathcal{R}\right)\\
&=\int_{N^-\times N^+}\left[f(u)g(v)- f(Jv)g(Ju)\right]dP_u(v) d\mu^-(u) \\
&=\alpha \mu^-(f)\mu^+(g) + (1-\alpha) \mu^-\left(f (g\circ \text{Ref})\right)\\
&\ \ \ \ \ \ \ \ \ \ \ \ \ - \left[\alpha \mu^-(g\circ J)\mu^+(f\circ J) + (1-\alpha) \mu^-\left((g\circ J) (f\circ J\circ \text{Ref})\right)\right].
\end{align*}
This last expression is seen to be  $0$ because 
$J_*\mu^\pm = \text{Ref}_*\mu^\pm = \mu^\mp $
and  the flip and reflection maps commute. (Note  that $\mu^+(f\circ J)=(J_*\mu^+)(f)=\mu^-(f)$.)

\subsection{Uniform ergodicity}\label{uniform ergodicity subsection}

Our present aim  is to study the Markov chain $X_n$ on state space $N^+$ with Markov transition kernel $\mathcal B$.
In the remainder of this section  we restrict ourselves to the case in
which $\mathcal{B}$ is induced by the Maxwell-Smoluchowski  reflection operator,
the billiard system is free of potential forces, and the bundle $W$ of thermal directions is all of $N$.

Before turning to the chain $X_n$, we  introduce a related Markov chain derived from the projection of $X_n$ onto the bundle of unit vectors in $N^+$.
Recall that $S \subset N $ denotes the bundle of unit vectors in $N$, $S^+=S \cap N^+$, and $S^+_q=S^+\cap N_q$.  The hemisphere bundle $S^+$  is invariant under the standard billiard map, which is defined as the composition of the translation map $\mathcal T$ and specular reflection. This is clear since $\mathcal T$ and the specular reflection map preserve the Riemannian norm.  The billiard measure $\sigma$  on $S^+$, introduced at the end of Section \ref{sec:main}, is the  measure invariant under the standard  billiard map,  obtained from the symplectic form as described, for example,  in \cite{cook}.  Recall that
$$d\sigma(q,u)= C\left\langle u, \mathbbm{n}_q\right\rangle\,  dV_q^{\text{\tiny S}}(u)\, dA(q), $$
where $A$ is the Riemannian $(n-1)$-dimensional volume measure on $\partial M$ and $V_q^{\text{\tiny S}}$ is the Riemannian $(n-1)$-dimensional volume measure on $S_q^+$. 

One property of the Maxwell-Smoluchowski reflection operator that should be highlighted is that it is {\em projective} according to the following definition.  We denote by $\Pi_q:N_q^+\setminus \{0\}\rightarrow S_q^+$ the projection map $\Pi_q u=u/|u|_q$ and by $\Pi$ the corresponding projection map from $N^+$ minus the zero section onto $S^+$.

\begin{definition}[Projective reflection operators]
We say that the reflection operator $P$ is {\em projective} if for all nonzero $x=(q, u)\in N^-$,   $\lambda>0$, and continuous $f:S_q^+\rightarrow \mathbb{R}$, the integral $ P_{(q,\lambda u)}(f\circ \Pi_q)$ does not depend on $\lambda$. 
\end{definition}

That the Maxwell-Smoluchowski model has this property is readily seen:
	$$ P_{(q, \lambda u)}(f\circ \Pi_q) =\alpha_q \mu_q(f\circ \Pi_q) + (1-\alpha_q) f(\Pi_q \text{Ref}_q \lambda u);$$
but $\Pi_q \text{Ref}_q \lambda u=\Pi_q \text{Ref}_q u$, so the left-side of the equation does not depend on $\lambda$.  Moreover, the associated billiard map $\mathcal{B}$ induces a map $\overline{\mathcal{B}}$ on $S^+$ as follows: given a continuous function $f$ on $S^+$,
	$$(\overline{\mathcal{B}} f)(\Pi x):=\mathcal{B}_x(f\circ\Pi).$$
The operator $\overline{\mathcal B}$ thus acts as the Markov transition kernel for the Markov chain $\overline X_n := \Pi(X_n)$.

In \cite{CPSV2009}, uniform ergodicity is established for a stochastic process related to $\overline X_n$, which informally can be described as follows.  A particle moves with constant speed inside the domain $M$.  Upon collision with the boundary, it is reflected in some random direction, not depending on the incoming direction, keeping the magnitude of its velocity constant.  The distribution of the random direction is given by the so-called \emph{Knudsen cosine law}, which is the velocity component of the billiard measure $\sigma$.  The particle then continues on to the next collision point, where it is again reflected in a random direction, independent of the previous collision, and this continues \emph{ad infinitum}.  The sequence of collision points, which forms a Markov chain with state space $\partial M$, is referred to as the \emph{Knudsen random walk}.

It is readily seen that the Knudsen random walk is precisely the process $\xi_n := \pi(\overline X_n)$ when $\alpha(q) \equiv 1$.  The following conditions on $\partial M$ are adapted from \cite{CPSV2009}.

\begin{assumption} \label{Lipschitz condition}
Suppose that $\partial M$ is an $(n-1)$-dimensional, almost everywhere continuously differentiable surface satisfying the following Lipschitz condition.  For each $q \in \partial M$, there exists $\epsilon_q > 0$, an affine isometry $\mathcal I_q : M \to \mathbb R^n$, and a function $f_q : \mathbb R^{n-1} \to \mathbb R$ such that
\begin{itemize}
	\item The function $f_q$ satisfies $f_q(0) = 0$ and the Lipschitz condition.  That is, there exists a constant $L_q > 0$ such that $|f_q(p)-f_q(p')| < L_q|p-p'|$ for all $p,p' \in \partial M$.
	\item The affine isometry satisfies $\mathcal I_q q = 0$ and $$\mathcal I_q(M \cap B(q,\epsilon_q)) = \{z \in \mathcal B(0,\epsilon_q) : z_n > f(z_1,\ldots,z_{n-1})\}.$$
\end{itemize}
\end{assumption}

\begin{theorem}[adapted from \cite{CPSV2009}]\label{Knudsen-ergodicity}  Suppose $\mathrm{diam}(M) < \infty$ and Assumption \ref{Lipschitz condition} holds.  Then the normalized Riemannian volume $\bar A$ on $\partial M$ is the unique stationary distribution for the Knudsen random walk.  Moreover, there exist constants $\beta_0,\beta_1$, independent of the distribution of $\xi_0$, such that $$||P(\xi_n \in \cdot)-\bar A||_v \leq \beta_0e^{-\beta_1 n},$$ where $||\cdot||_v$ is the total variation norm.
\end{theorem}

In what follows, we extend the result above to show that the Markov chains $\overline X_n$ and $X_n$ are uniformly ergodic.  Moreover, we give explicit expressions for the stationary measures of these chains.  Before stating the theorem, we first establish some notation. Recall that we assume that $\partial M$ is partitioned into $N$ components $\Gamma_i$.  Moreover $T(q) \equiv T_i$ and $\alpha(q) \equiv \alpha_i$ for all $q \in \Gamma_i$.  Let $$p_{ij} = P(\xi_{n+1} \in \Gamma_j \mid \xi_n \in \Gamma_i)$$ be the one-step transition probability of the Knudsen random walk between components of $\partial M$.  Let $A_i = A(\Gamma_i)$ and $\bar A_i = \bar A(\Gamma_i)$ be the volume measure and normalized volume measure of the boundary components $\Gamma_i$ respectively.  Next, note that since the temperature is constant on boundary components, by identifying $N_q^+$ with the upper half space $\mathbb R^n_+ = \{x \in \mathbb R^n : x \cdot e_n > 0\}$, we can define $\mu_i := \mu_q \in \mathcal P(N_q^+)$ for $q \in \Gamma_i$ to be the Maxwellian associated to $\Gamma_i$.  Moreover, it is readily apparent that $\mu_i$ can be disintegrated using polar coordinates as a product measure on $S_q^+\times(0,\infty)$.  One can check that the component on $S_q^+$ is $d\sigma_q(u)=C_q\left\langle u, \mathbbm{n}_q\right\rangle\,  dV_q^{\text{\tiny S}}(u)$, the Knudsen cosine law referred to above.  We denote the component on $(0,\infty)$, the speed component, as $\mu_i^\mathrm s$: $$d\mu_i(u,r) = d\sigma_q(u)d\mu_i^\mathrm s(r).$$  Finally, let $Q$ be the $N \times N$ matrix and $\pi$ the $N$-dimensional vector where $$Q_{ij} = (1-\alpha_i)p_{ij} A_i/A_j, \quad \pi_i = \bar A_i\alpha_i\mu_i^\mathrm s$$

\begin{theorem}\label{uniform ergodicity}
Suppose $\mathrm{diam}(M) < \infty$ and that Assumption \ref{Lipschitz condition} holds.  Then
\begin{enumerate}
	\item The billiard measure $\sigma$ is stationary for $\overline X_n$.
	
	\item There exists a unique stationary distribution $\nu$ for $X_n$.  Moreover, there exist constants $b_0,b_1$, independent of the distribution of $X_0$, such that $$||P(X_n \in \cdot)-\nu||_v \leq b_0e^{-b_1 n},$$ where $||\cdot||_v$ is the total variation norm.  That is, the chain $X_n$ is uniformly ergodic.
	
	\item When the boundary temperatures $T_1,\ldots,T_N$ are equal, say to a constant $T_0 > 0$, the stationary distribution $\nu$ is given by $$d\nu(q,v) = \rho_q(v)dV_q(v)d\bar A(q)$$ where $\rho_q(v)$ is the Maxwellian density given in (\ref{maxwellian}) with constant temperature $T(q) \equiv T_0$.
	
	\item When the boundary temperatures are not equal, the stationary distribution $\nu$, expressed in polar coordinates as a measure on $S^+ \times (0,\infty)$ is given by $d\nu(q,u,r) = d\nu^{\mathrm s}_q(r)d\sigma(q,u)$, where $\nu^\mathrm{s}_q$ is the stationary distribution of the speed after collision with boundary point $q$.  The measure $\nu_q^\mathrm s$ is constant on components $\Gamma_i \subset \partial M$, so we let $\nu_i^\mathrm s := \nu_q^\mathrm s$ for $q \in \Gamma_i$.  Letting $\nu^\mathrm s$ be the $N$-dimensional vector with components $\nu_i^\mathrm s$, we have that $$\nu^\mathrm s = (I - Q)^{-1}\pi.$$
\end{enumerate}
\end{theorem}

\begin{proof}
Let $\mu$ be the measure derived  from the Maxwellians $\mu_q$ as indicated just prior to the statement of Proposition \ref{entropy_reciprocity}.  Note that $\overline\mu := \Pi_*\mu = \sigma$.  Moreover, it is a consequence of the reciprocity property of $P$ that the following time reversibility condition holds: for all $x,y \in N^+$, $$d\overline{\mu}(\overline{x})d\overline{\mathcal{B}}_{\overline{x}}(\overline{y})=d\overline{\mu}(\overline{y})d\overline{\mathcal{B}}_{\overline{y}}(\overline{x}),$$ where $\overline x := \Pi x$.  It then follows that $\sigma \overline{\mathcal B} = \sigma$ and the first part of the theorem is proved.  Moreover, the same argument can be used to show that the third part of the theorem holds.

For the second part of the theorem, we show that the chain $X_n$ is uniformly ergodic by showing that the state space $N^+$ is small in the following sense: there exists $m \in \mathbb Z_+$ and a nontrivial measure $\phi$ on $N^+$ (which is not necessarily a probability measure) such that for all $x \in N^+$ and measurable sets $A \subseteq N^+$, we have that $\mathcal B^m_x(A) \geq \phi(A)$ (see Theorem 16.0.2 of \cite{MT2009}).  By Theorem \ref{Knudsen-ergodicity}, there exists $m \in \mathbb Z_+$ and a nontrivial measure $\overline\phi$ on $S^+$ so that the aforementioned condition for uniform ergodicity holds for the $\overline X_n$ chain.  We can then construct the measure $\phi$ as follows: $$\phi(A) = \int_{\Pi(A)}\int_{L_x(A)} \rho^{\text{\tiny s}}_q(|v|)\,d|v|\,d\overline\phi(\overline x),$$ where $L_x(A) := \{|v| : \overline x = (q,v/|v|) \in \Pi(A)\}$ and $\rho_q^{\text{\tiny s}}(|v|) = \int_{S_q^+} |v|^{n-1}\rho_q(v) dV_q^{\text{\tiny S}}(v/|v|)$.  Note that $$B_x^m(A) = \int_{\Pi(A)}\int_{L_{x'}(A)} \rho_{q'}^{\text{\tiny s}}(|v'|)d|v'|d\overline{\mathcal B}_{\overline x}^m(\overline{x'}) \geq \int_{\Pi(A)}\int_{L_{x'}(A)} \rho_{q'}^{\text{\tiny s}}(|v'|)d|v'|d\overline\phi(\overline{x'}) = \phi(A).$$

For the final part of the theorem, let $\nu_0$ be the initial distribution of the Markov chain and define $\nu_k = \nu_0 \mathcal B^k$ to be the distribution of $X_k$.  Further, note that the state space $N^+$ can be partitioned into components $N_i^+ = \{x = (q,u) \in N^+ : q \in \Gamma_i\}$ and we let $\nu_{k,i}$ denote the restriction of $\nu_k$ to $N_i^+$.  Note that the stationary measure for $X_n$ must assign probability $\bar A_i$ to component $N_i^+$ since the normalized area measure on the boundary is stationary for $\overline X_n$.  Moreover, the relation $\nu_{k+1} = \nu_k \mathcal B$ amounts to the system $$\nu_{k+1,i} = \bar A_i \alpha_i \mu_i + \sum_{j=1}^N (1-\alpha_i) p_{ij} A_i/A_j \nu_{k,j},$$ for $i = 1,\ldots, N$.  Expressed in matrix form, this yields $\nu_{k+1} = \pi + Q\nu_k$, which implies $$\nu_k = Q^k\nu_0 + \sum_{n=0}^{k-1} Q^n\pi.$$  It is straightforward to check that this converges in total variation to $(I-Q)^{-1} \pi$.

\end{proof}

\subsection{Entropy production rate formula}

With an explicit expression for the stationary measure in hand, we are now ready to give a formula for the entropy production rate.  It is clear that the hypotheses of Theorem \ref{heat entropy} are satisfied by the stationary measure given in Theorem \ref{uniform ergodicity}.  The main work will be in computing $\nu_q^+(E_0)-\nu_q^-(E_0)$, where $E_0$ is the kinetic energy function.

\begin{proposition}\label{entropy formula proposition} The entropy production rate in the multi-temperature Maxwell-Smoluchowski system is given by
\begin{equation}\label{entropy formula multitemp}
e_p = -\sum_{j =1}^N \frac{\nu_j^+(E_0) - \nu_j^-(E_0)}{\kappa T_j} = -\sum_{i,j =1}^N \frac{\nu_j^+(E_0) - \nu_i^+(E_0)}{\kappa T_j} p_{ij}\frac{A_i}{A_j}
\end{equation}
where $\nu^+$ is the stationary measure $\nu$ given in Theorem \ref{uniform ergodicity}, $\nu^- = \mathcal T_* \nu$, and $\nu_j^\pm$ is the restriction of the stationary measure $\nu^\pm$ to component $N_j^\pm = \{(q,u) \in N^\pm : q \in \Gamma_i\}$.  
\end{proposition}

The factor $\nu_j^+(E_0)-\nu_i^+(E_0)$, which quantifies the stationary state energy exchange between wall $j$ and wall $i$, can be expressed more explicitly in terms the temperature gradient between wall $j$ and wall $i$, but we leave this to some explicit examples of tables below.

\begin{proof}
Let $E_0(q,u) = \frac 12 m|u|_q^2$.  Observe that
\begin{align*}
\int_{\partial M} \frac{\nu_q^\pm(E_0)}{\kappa T(q)}\,d\bar A(q) 
&= \sum_{j = 1}^N \int_{\Gamma_j} \frac 1{\kappa T(q)}\left(\int_{N_q^+} E_0(q,u)\, d\nu_q^\pm(u) \right)\, d\bar A(q) \\
&= \sum_{j = 1}^N \frac{1}{\kappa T_j} \int_{\Gamma_j}\int_{N_q^+} E_0(q,u) d\nu^\pm(q,u) \\
&= \sum_{j = 1}^N \frac{\nu_j^\pm(E_0)}{\kappa T_j} 
\end{align*}
Next, note that $\nu_j^-(E_0) = \sum_{i = 1}^N p_{ij} A_i/A_j \nu_i^+(E_0)$.  Using the entropy production rate formula in Theorem \ref{heat entropy} this gives
\begin{align*}
e_p &= -\sum_{j = 1}^N \frac{\nu_j^+(E_0) - \nu_j^-(E_0)}{\kappa T_j} \\
&= -\sum_{j = 1}^N \frac{\nu_j^+(E_0) - \sum_{i = 1}^Np_{ij}A_i/A_j \nu_i^+(E_0)}{\kappa T_j} \\
&= -\sum_{i,j =1}^N \frac{\nu_j^+(E_0) - \nu_i^+(E_0)}{\kappa T_j} p_{ij}\frac{A_i}{A_j},
\end{align*}
where the final equality follows because $\sum_{i = 1}^Np_{ij}A_i/A_j = P(\xi_n \in \Gamma_i \mid \xi_{n+1} \in \Gamma_j) = 1$ for each $j = 1,\ldots, N$, when $\xi_n$ is the stationary Knudsen random walk. 
\end{proof}

\subsection{Examples}\label{subsec:Examples}

\subsubsection{The two-plates system}

     We illustrate the formula giving  the  entropy production rate for  the elementary system indicated in Figure \ref{two_plates}. It consists of a particle that bounces back and forth between two parallel plates kept at temperatures $T_1$ and $T_2$. For the reflection operator we adopt the Maxwell-Smoluchowski model with parameters $\alpha_1$ and $\alpha_2$. Thus at any point $q$ of plate $i$, the post-collision velocity of the colliding particle
     has the Maxwell-Boltzmann distribution with temperature $T_i$ with probability  $\alpha_i$, and with probability $1-\alpha_i$ the particle reflects specularly.

      The manifold is taken to be the product of a flat torus and an interval, $M=\mathbb{T}^2\times [0,l]$, where $l$ is the distance between the two (torus) plates, which comprise the two connected components of the boundary $\partial M=\mathbb{T}^2\times \{0,l\}$. Let the index $1$ be associated with the left-side plate and $2$ with right-side plate.  The phase space $N$ is the union of components $N^\pm_i$, where $i\in \{1,2\}$.  Each $N^\pm_i$ is identified with 
     $\mathbb{R}^3_+=\{u\in \mathbb{R}^3: u\cdot \mathbbm{n}>0\}$, where $\mathbbm{n}$ is the normal vector to plate $1$.
      We indicate by $\bar{i}$ the index opposite to $i$, so $\bar{1}=2$ and $\bar{2}=1$. Then the translation map assumes the form $\mathcal{T}(i,u)=(-\bar{i},u)$.  We write  $f_i(u)=f(i,u)$ where
     $u\in \mathbb{R}^3_+$

The billiard map applied to a function $f$ on $N$ has the form
$$ (\mathcal{B}f)(i,u)=P_{-\bar{i}, u}(f) = \alpha_{\bar{i}}\mu_{\bar{i}}(f) + (1-\alpha_{\bar{i}}) f_{\bar{i}}(u).$$
Because the billiard particle alternates between the two plates, a stationary measure for $\mathcal{B}$ must assign equal probabilities 
for each plate: $\nu(N_i)=1/2$. We can then restrict attention to the two-step chain describing the sequence of returns to a plate. 
Let  $\nu_0$ be an initial distribution for the Markov chain and define $\nu_k=\mathcal{B}^k\nu_0$. Let $\nu_{k, i}$ be the restriction of these measures to plate $i$.  Then the equation $\nu_{k+1}=\nu_k \mathcal{B}$ amounts to the system
\begin{align*}
\nu_{k+1,1}&= \frac12 \alpha_1\mu_1 + (1-\alpha_1)\nu_{k,2}\\
\nu_{k+1,2}&=\frac12 \alpha_2 \mu_2 +(1-\alpha_2)\nu_{k,1}
\end{align*}

\begin{figure}
\begin{center}
\includegraphics[width=0.4\textwidth]{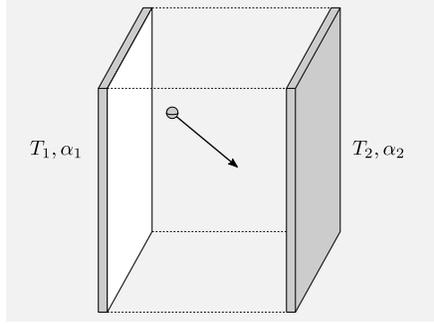}\ \ 
\caption{{\small A particle of mass $m$ bounces back and forth between two plates kept at temperatures $T_1$ and $T_2$. For the reflection operator we use the Maxwell-Smoluchowski model with probabilities of diffuse reflection $\alpha_1$ and $\alpha_2$. }}
\label{two_plates}
\end{center}
\end{figure}

 Writing $$\nu_k=\left(\begin{array}{c}\nu_{k,1} \\\nu_{k,2}\end{array}\right),\ \ Q=\left(\begin{array}{cc}0 & 1-\alpha_1 \\1-\alpha_2 & 0\end{array}\right),\ \ \pi=\frac12\left(\begin{array}{c}\alpha_1\mu_1 \\\alpha_2\mu_2\end{array}\right),$$ equation $\nu_{k+1}=\nu_k \mathcal{B}$ becomes $\nu_{k+1} = \nu_kQ + \pi$, from which we easily obtain
 $$ \nu_{2k} =\gamma^k \nu_0 + \frac{1-\gamma^k}{1-\gamma\ } (I+Q)\pi$$
 where $\gamma=(1-\alpha_1)(1-\alpha_2)$. 
  We conclude that the process converges to a stationary process having stationary distribution $\nu=\nu_1+\nu_2$ where
$$
\nu_1 =\frac{\alpha_1}{2c}\mu_1 + \frac{\alpha_2(1-\alpha_1)}{2c}\mu_2, \ \  
\nu_2= \frac{\alpha_1(1-\alpha_2)}{2c}\mu_1 + \frac{\alpha_2}{2c}\mu_2
$$
  where $c=1-\gamma=1-(1-\alpha_1)(1-\alpha_2)$.  Recall that $\nu^-=\mathcal{T}_*\nu$. Thus $\nu_i^-=\nu_{\bar{i}}$:
  $$\nu^-_1 =\frac{\alpha_1(1-\alpha_2)}{2c}\mu_1 + \frac{\alpha_2}{2c}\mu_2, \ \ 
\nu^-_2=  \frac{\alpha_1}{2c}\mu_1 + \frac{\alpha_2(1-\alpha_1)}{2c}\mu_2.
$$

The stationary measure clearly satisfies the conditions of Theorem \ref{heat entropy}.
Moreover
 $$\nu_1^+-\nu_1^- = -(\nu_2^+ - \nu_2^-) = \frac{\alpha_1\alpha_2}{2c}(\mu_1-\mu_2). $$
 A simple computation gives $\mu_i(E_0)=kT_i$.   The entropy formula given in Theorem \ref{heat entropy} then yields the result:
 $$ e_p = -\frac{\alpha_1\alpha_2}{2\left[1-(1-\alpha_1)(1-\alpha_2)\right]} (T_1- T_2)\left(\frac1{T_1}-\frac1{T_2}\right).$$
 This expression is clearly non-negative. 
 Let us denote by $Q_i$ the expected change in energy of the billiard particle at a collision with plate $i$ in the stationary regime.
  Then $$Q_i= \nu_i^+(E_0) - \nu_i^-(E_0)= -(-1)^i\frac{\alpha_1\alpha_2}{2c}\left(\kappa T_1-\kappa T_2\right). $$     Thus if $T_1>T_2$,
  the particle takes, on average, an amount $Q:=\frac{\alpha_1\alpha_2}{2c}\left(\kappa T_1-\kappa T_2\right)$ of energy from plate $1$  and 
  transfers it to plate $2$ per collision.     Combining with the expression for $e_p$ results in
  $$e_p = \frac{Q}{\kappa T_2} - \frac{Q}{\kappa T_1}. $$  
 It should be noticed     that this expression is independent of the nature of the thermostat model assumed for the plates. On the other hand,
 the above expression showing  that  $Q$ is proportional to the temperature difference depends on the choice  of  model. In fact, the coefficient $\frac{\alpha_1\alpha_2}{2\left[1-(1-\alpha_1)(1-\alpha_2)\right]}$ tells how fast the system transfers energy from hot to cold plate, and may thus be interpreted as a heat conductivity parameter  (measured per pair of collisions rather than time between collisions; the latter  may also be calculated from the stationary measure).  
 
\subsubsection{A three-temperature system}  

Next, we wish to express the entropy production rate for systems with more than two temperatures.  As a prototypical example we take $\partial M$ to be an equilateral triangle where boundary component $\Gamma_i \subset \partial M$ is equipped with parameters $T_i$ and $\alpha_i$ for $i = 1, 2, 3$.  Note that the calculations and formulas to be shown can be generalized to more general polygons, but we restrict our attention to the equilateral triangle for the sake of simplicity.

Following the notation in Subsection \ref{uniform ergodicity subsection} and Equation (\ref{entropy formula multitemp}), the side lengths $A_i$ are equal for all $i$ and the normalized side length $\bar A_i = 1/3$.  Moreover, by symmetry $p_{ij} = (1-\delta_{ij})/2$ for $i,j = 1,2, 3$, where $\delta_{ij}$ denotes the Kronecker delta.  Using Theorem \ref{uniform ergodicity}, $$\nu_j^+ = \sum_{k = 1}^3 c_{jk}\mu_k^+, \quad \mbox{where $c_{jk} = \alpha_j(I-Q)^{-1}_{jk}/3$,}$$ and $Q_{jk} = (1-\alpha_j)(1-\delta_{jk})/2$.  As in the two-plates example, we let $Q_i$ be the expected change in energy of the billiard particle at a collision with boundary component $i$ in the stationary regime.  Moreover, let $\overline Q_{ij}$ denote the average amount of energy from component $i$ which is transferred to component $j$ per collision.

A tedious but straightforward calculation gives
\begin{align*}
Q_i = \nu_i^+(E_0) - \nu_i^-(E_0) &= \frac{\alpha_i}{6}\left(1 - \frac{\alpha_{i'}}{2}\left(1-\alpha_{i''}\right)\right)\mu_{i}(E_{0}) - \frac{\alpha_{i}\alpha_{i'}}{4}\left(1-\frac{\alpha_{i''}}{3}\right)\mu_{i'}(E_0) \\
&\quad \quad \quad \quad + \frac{\alpha_i}{6}\left(1 - \frac{\alpha_{i''}}{2}\left(1-\alpha_{i'}\right)\right)\mu_{i}(E_{0}) - \frac{\alpha_{i}\alpha_{i''}}{4}\left(1-\frac{\alpha_{i'}}{3}\right)\mu_{i''}(E_0) \\
&= \overline Q_{ii'} + \overline Q_{ii''}
\end{align*}
 for $i = 1,2,3$, where $i' = i+1\mod 3$ and $i'' = i+2\mod 3$.  Note that another simple computation yields that $\mu_i(E_0) = \frac{3}{2^{3/2}}\kappa T_i$ for a two-dimensional billiard domain.  Using the entropy formula, we have that $$e_p = -\sum_{i = 1}^3 \frac{Q_i}{\kappa T_i}.$$
 
We have expressed $Q_i$ as the sum of two differences $\overline Q_{ij}$ in order to emphasize that the expected change in energy involves a transfer of energy among pairs of boundary components.  While the formulas are a bit more complicated in the case of multiple temperatures, they nevertheless capture the general qualitative property that the entropy production scales with the square of the temperature difference.  On the other hand, the coefficients on the terms $\mu_i(E_0)$ are specific to the model and express how fast the system transfers energy among boundary components.
 
\subsubsection{A circular chamber system}

We conclude this section with a numerical example to demonstrate how geometric features of the billiard  table, as opposed to features of the collision model, can influence the entropy production rate.

\begin{figure}
\begin{center}
\includegraphics[width=0.55\textwidth]{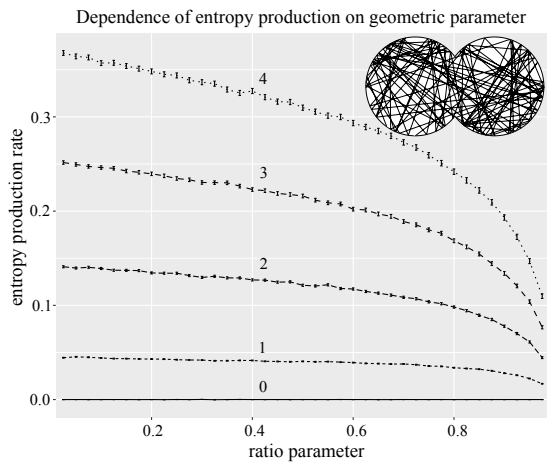}\ \ 
\caption{{\small Entropy production for a billiard system whose billiard domain is formed by the union of two discs of equal radius $r$ whose centers are $a$ units apart.  
The ratio parameter is $a/2r$ and the number above each graph is the temperature difference $T_2-T_1$. The vertical bars indicate 95\% confidence intervals.}}
\label{dependence-on-ratio}
\end{center}
\end{figure} 

The inset in Figure \ref{dependence-on-ratio} shows the billiard table of interest. It  consists of two overlapping discs of radius $r$ with  centers at  a distance $a$ apart. We call $a/2r$ the {\em ratio parameter}.
The  boundary of the table   is the union of two symmetric arcs of circles kept at  constant temperatures  $T_1$ and $T_2$  and equipped with 
 the Maxwell-Smoluchowski collision model.
Note that when $a/2r = 0$, the two discs coincide and the boundary segments    are the right and left hemispheres; and when $a/2r = 1$, the discs are in tangential contact. 
 We are interested in the  changes in  entropy production rate $e_p$ due to varying
 the ratio parameter over the interval $(0,1)$. 
 
 Using the formula for the stationary measure in Theorem \ref{uniform ergodicity} and the formula for entropy production in Proposition \ref{entropy formula proposition}, the rate $e_p$ is determined by the transition probabilities $p_{ij}$ between boundary segments of $\partial M$.  These probabilities  are then estimated  
through numerical simulation of the billiard dynamic by doing a large sampling from the phase space of initial conditions and keeping a tally of one-step transitions between boundary segments. That is, $p_{ij}$ is estimated to be the proportion of trajectories from the sample that start at boundary segment $i$, i.e. circular arc $i$, and then have a first collision somewhere along boundary segment $j$, where $i, j = 1, 2$.

The five   graphs of Figure 
\ref{dependence-on-ratio} correspond to  
  $T_2-T_1 = 0,\ldots,4$. (These  values   are  indicated on top of each graph.)
    Each curve in the graph gives values of $e_p$ for 40 values of the ratio parameter.  
An interesting feature exhibited by the graphs in  Figure \ref{dependence-on-ratio}   is the sharp transition in the rate of decay of $e_p$
past a value of $a/2r$ roughly between $0.8$ and $1.0$ across different values of $T_2-T_1$.

 \section{Work production}\label{sec:Efficiency and work production}
 We view this paper's focus on  entropy production  and the second law      as only  a first step in a broader investigation of the stochastic thermodynamics for random billiards. It is natural to ask whether these systems can be used to explore a wider range of classical thermodynamic phenomena. In this final section we wish to show  how the issue of work production can very naturally be brought into the general picture by exploring, mainly numerically, a   random billiard model of a heat engine.

\begin{figure}
\begin{center}
\includegraphics[width=0.3\textwidth]{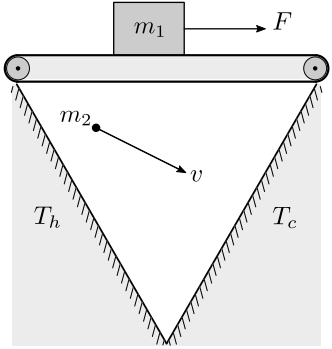}\ \ 
\caption{{\small A simple random billiard heat engine.  The difference in side wall temperatures allows the system to do work against an external force.}}
\label{moving_belt}
\end{center}
\end{figure} 
 
 The  system we use to illustrate this point is shown in Figure \ref{moving_belt}. It is only one extremely simple example of  a general class of models that we call {\em thermophoretic  engines}, to be considered more systematically in a future study. For now,  this  example   will serve  to show  the possibilities offered by random billiard models in stochastic thermodynamics beyond the more restricted purview of the present paper.
 
 In Figure \ref{moving_belt},  a point mass $m_2$ can move freely inside an equilateral triangular domain whose sides are of two types: two of them can thermally interact   with the particle through
  the Maxwell-Smoluchowski thermostat model with  temperatures $T_h$ and $T_c$; the  third side at the top of the triangle can slide without friction as a kind of conveyor belt. This side has its own mass, $m_1$. A constant force $F$ is applied to $m_1$. Among the many possible choices for the model of interaction between the two masses we choose here the simplest, and not necessarily the most physically natural, that allows for an exchange of momentum:  when $m_2$ collides with the top side of the triangle the perpendicular component of its velocity changes sign and the horizontal component, together with the velocity of the belt,  change as
  in the two-masses example of Figure \ref{twomasses}. That is, the interaction is as
   if $m_1$ and $m_2$ undergo a frontal elastic collision in dimension $1$. 
 
 We first describe   how this system fits the definition of a random billiard.  Let   $L$ denote the length of the conveyor belt  and let  $\mathcal{S}$ be the interior of the triangle. The configuration manifold is then the $3$-dimensional space $M=\mathcal{S}\times  \mathbb{R}/ L\mathbb{Z}$. As in the two-masses system of Section \ref{sec:Second Law of Thermodynamics}, it is convenient to rescale the position coordinates of $L$ and $\mathcal{S}$ so as to make the total kinetic energy   proportional to the square Euclidean norm of the combined velocity vector of the two masses. In other words, if $(x,y)$ are the coordinates on the plane of $\mathcal{S}$ and $z$ parametrizes the position of the conveyor belt, we set
 $$x_1= \sqrt{\frac{m_1}{m}}z, \ \ x_2=  \sqrt{\frac{m_2}{m}}x, \ \ x_3=  \sqrt{\frac{m_2}{m}}y, $$
 where $m=m_1+m_2$.  Let us also denote $\gamma:=\sqrt{m_2/m_1}$.
 We  introduce a positive orthonormal frame on the boundary of $M$ denoted $(e_1, e_2, e_3)$ where $e_1$ points in the direction of the axis  $x_1$, $e_2$ is tangent to the boundary of $M$, and $e_3$ is perpendicular to the boundary, pointing to the interior of $M$. 
 Let $v$ and $V$ denote   velocity vectors before and after a collision, respectively, expressed in the coordinate system given by $x=(x_1, x_2, x_3)$.
 Collisions with the stationary sides of the triangle are assumed to satisfy the reciprocity condition such that at each $x\in \partial M$ the subspace of thermal directions, in the sense of Definition \ref{thermal direction},
is perpendicular to $e_1(x)$. Collisions with the sliding top side  are deterministic. Conservation of energy and momentum implies that $V=Cv$ where the collision map $C\in SO(3)$ is the orthogonal involution given by
$$C=\left(\begin{array}{ccr}\frac{1-\gamma^2}{1+\gamma^2} &\ \, \frac{2\gamma}{1+\gamma^2} & 0 \\[6pt]    \frac{2\gamma}{1+\gamma^2} & -\frac{1-\gamma^2}{1+\gamma^2} & 0 \\[6pt]0 &\  0 & -1\end{array}\right). $$
Note that the restriction of $C$ to $T_x\partial M$ is a reflection. The orthogonal component  of $v$ in the subspace spanned by $\gamma e_1 - e_2$  and $e_3$ changes sign, whereas 
the component parallel to $e_1+\gamma e_2$ remains unchanged.

\begin{figure}
\begin{center}
\includegraphics[width=0.5\textwidth]{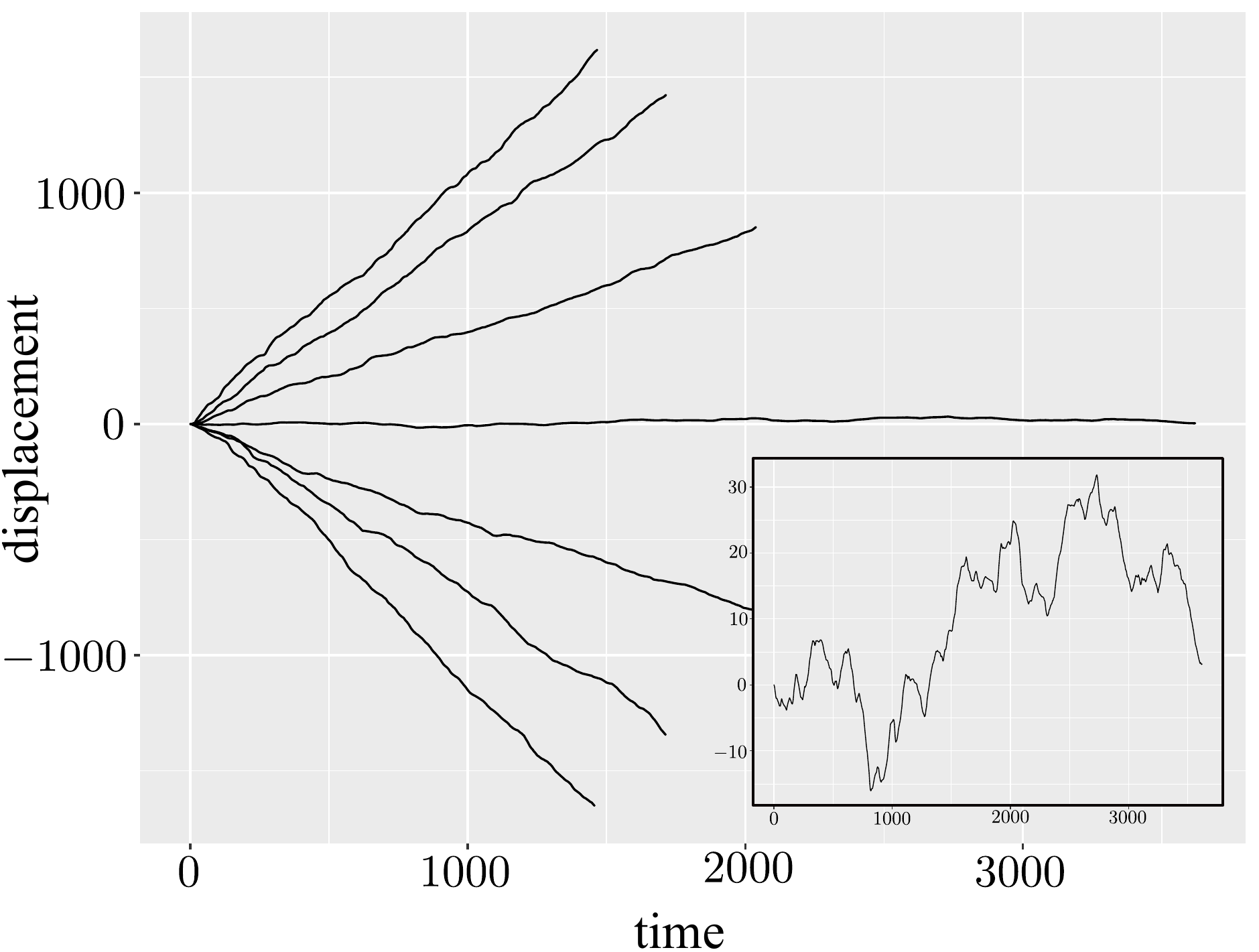}\ \ 
\caption{{\small When $F=0$ the conveyor belt steadily drifts counterclockwise when $T_h-T_c>0$, and clockwise when the temperatures are reversed.
The temperatures for the top  $4$ graphs are $T_c=1$ and $T_h=1, 10, 25, 50$, and $T_c, T_h$ are reversed for the $3$ lower graphs. The inset is the same as the graph for $T_c=T_h=1$ but in a finer scale so that its stochastic character  is more clearly apparent.  }}
\label{translation-vs-time}
\end{center}
\end{figure} 

 Observe how  the particle of mass $m_2$ imparts momentum to the sliding side of the triangle
  as it mediates the energy transfer between the two thermal sides.  In the absence of the force $F$, we should expect the conveyor  belt to move with a steady drift conterclockwise if $T_h>T_c$, and move in the opposite direction when the temperatures are reversed. We should also expect, for some range of values of $F$ and $T_h-T_c$, to observe work being produced against the force. This is in fact what is obtained by the numerical  simulation. Figure \ref{translation-vs-time} shows the motion of the conveyor belt as a function of time for different values of $T_h-T_c$.

  As expected, the stochastic motion of the conveyor belt exhibits a steady drift compatible with an energy flow from the hot side to cold. The middle graph among the seven described in
  Figure \ref{translation-vs-time}  corresponds to $T_h=T_c$ and it is also shown in a different scale in the inset figure.

 If we now impose a constant external force $F$ on $m_1$, a fraction  of the energy flow between the thermal sides is converted into work against $F$. 
 More precisely, 
   at each time $t > 0$, let $Q_h(t)$ be the total amount of energy transferred to the system since $t = 0$ due to collisions between the particle and the hot wall.  Let $Q_c(t)$ be the energy transferred due to collisions with the cold wall. Let $x_w(t)$ be the signed displacement  from the initial position $x_w(0)$ of the mass $m_1$ at time $t$ and let $W(t) = (x_w(t) - x_w(0))F$ be the work done by the system.  The Carnot mean efficiency of the engine up to time   $t$   is then   $$\epsilon_t(T_h, T_c) = -\frac{W(t)}{Q_h(t)}.$$  
Let us now take  into account, in standard fashion, energy conservation.
   The internal energy of the system at time $t$ is $E(t) = E_w(t) + E_p(t)$, which consists  of the kinetic energy $E_w(t)$ of mass $m_1$ and the kinetic energy $E_p(t)$ of the particle $m_2$.  Then conservation of energy, or the First Law of Thermodynamics, gives  at each $t$ $$E(t) - E(0) = Q_h(t) + Q_c(t) + W(t).$$

\begin{figure}
\begin{center}
\includegraphics[width=0.45\textwidth]{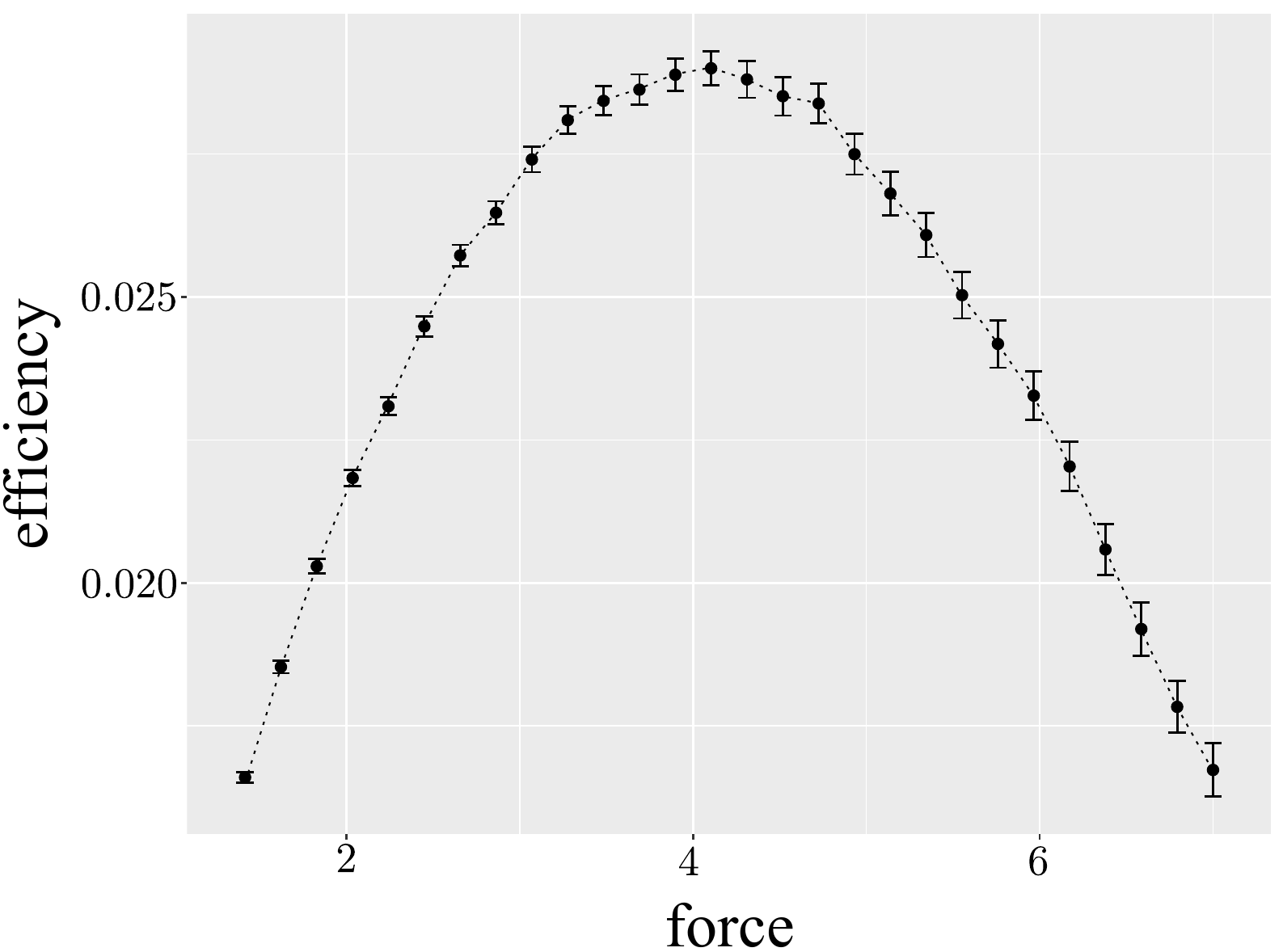}\ \ 
\caption{{\small Efficiency of the billiard heat engine
as a function of the force acting on the sliding
wall. }}
\label{efficiency-vs-force}
\end{center}
\end{figure} 

As $Q_h(t)$ is the accumulated   energy that flows through the hot side up to time $t$, and since $E(t)$ can be expected to have a finite mean value over time, 
it should be the case (and is seen experimentally) that $(E(t) - E(0))/Q_h(t) \to 0$
for large values of $t$.   This then yields the steady state expression of efficiency  given by $$\overline{\epsilon}_t(T_h, T_c) = 1 + Q_c(t)/Q_h(t).$$  Figure \ref{efficiency-vs-force} shows the characteristic curve of mean efficiency   as a function  of  $F$.
The simulation evaluates, for each value of the force,  the efficiency over $5000$ sample runs, with each run of length corresponding to  $1000$ collision events.  The parameters are $T_h = 50, T_c = 1,  m_1 = 1000, m_2 = 1.$  The vertical bars indicate $95\%$ confidence intervals.

\bibliographystyle{abbrv}
\bibliography{Entropy_production}

\end{document}